\documentclass[12pt, draftclsnofoot, onecolumn]{IEEEtran}
\makeatletter
\def\ps@headings{%
\def\@oddhead{\mbox{}\scriptsize\rightmark \hfil \thepage}%
\def\@evenhead{\scriptsize\thepage \hfil \leftmark\mbox{}}%
\def\@oddfoot{}%
\def\@evenfoot{}}
\usepackage[english]{babel}
\usepackage[margin=1in]{geometry}
\usepackage{amsmath,amsthm,amssymb,scrextend}
\allowdisplaybreaks[4]
\usepackage{graphicx}
\usepackage{fancyhdr}
\usepackage{cases}
\usepackage{framed}
\usepackage{enumerate}
\usepackage{enumitem}
\usepackage{indentfirst}

\usepackage{amssymb,amsmath,color,graphicx}
\usepackage{subfigure}
\usepackage{cite}
\usepackage{verbatim}
\usepackage{url}
\usepackage{graphicx,epstopdf}
\usepackage{bm}
\usepackage{textcomp}

\usepackage{amsmath}
\usepackage{xparse}
\usepackage{caption}

\makeatletter
\NewDocumentCommand{\Exp}{m g}{%
  \IfNoValueTF{#2}%
  {\relax\if@display
        \mathbb{E}\!\left[{#1}\right]%
    \else
        \mathbb{E}[{#1}]%
    \fi}%
  {\relax\if@display
        \mathbb{E}\!\left[{#1}\,\middle|\,{#2}\right]%
    \else
        \mathbb{E}[{#1}\,|\,{#2}]%
    \fi}%
}
\makeatother

\usepackage[linesnumbered,ruled,vlined]{algorithm2e}
\usepackage{algpseudocode}

\DeclareMathOperator*{\argmin}{arg\,min}

\newtheorem{theorem}{Theorem}
\newtheorem{assumption}{Assumption}

\newtheorem{corollary}{Corollary}

\newtheorem{lemma}{Lemma}

\newtheorem{proposition}{Proposition}
\newtheorem{remark}{Remark}

\newcommand{\mc}[1]{\mathcal{#1}}
\newcommand{\bs}[1]{\boldsymbol{#1}}

\begin{document}
\captionsetup[figure]{labelfont={footnotesize},labelformat={simple},labelsep=period,name={Fig.}}
\title{Integrated Sensing, Computation, and Communication for UAV-assisted Federated Edge Learning}
\author{Yao Tang, Guangxu Zhu, Wei Xu, Man Hon Cheung, Tat-Ming Lok, and Shuguang Cui
\thanks{Yao Tang and Tat M. Lok are with the Department of Information Engineering, the Chinese University of Hong Kong (CUHK), Hong Kong (e-mail: \{ty018, tmlok\}@ie.cuhk.edu.hk).
Guangxu Zhu is with Shenzhen Research Institute of Big Data, Shenzhen, China (e-mail: gxzhu@sribd.cn).
Wei Xu is with the National Mobile Communications Research Lab, Southeast University, Nanjing 210096, China (e-mail: wxu@seu.edu.cn).
Man Hon Cheung is with the Department of Computer Science, City University of Hong Kong, Hong Kong (e-mail: mhcheung@cityu.edu.hk).
Shuguang Cui is with the School of Science and Engineering (SSE), the Future Network of Intelligence Institute (FNii), and the Guangdong Provincial Key Laboratory of Future Networks of Intelligence, The Chinese University of Hong Kong (Shenzhen), 
Shenzhen, China. 
He is also with Peng Cheng Laboratory (e-mail: shuguangcui@cuhk.edu.cn).
}
}

\maketitle

\vspace{-0.55cm}

\pagestyle{headings} 
\thispagestyle{empty} 

\begin{abstract}
Federated edge learning (FEEL) enables privacy-preserving model training through periodic communication between edge devices and the server. 
Unmanned Aerial Vehicle (UAV)-mounted edge devices are particularly advantageous for FEEL due to their flexibility and mobility in efficient data collection. 
In UAV-assisted FEEL, sensing, computation, and communication are coupled and compete for limited onboard resources, and UAV deployment also affects sensing and communication performance. Therefore, the joint design of UAV deployment and resource allocation is crucial to achieving the optimal training performance.
In this paper, we address the problem of joint UAV deployment design and resource allocation for FEEL via a concrete case study of human motion recognition based on wireless sensing.
We first analyze the impact of UAV deployment on the sensing quality and identify a threshold value for the sensing elevation angle that guarantees a satisfactory quality of data samples.
Due to the non-ideal sensing channels, we consider the probabilistic sensing model, where the successful sensing probability of each UAV is determined by its position.
Then, we derive the upper bound of the FEEL training loss as a function of the sensing probability. 
Theoretical results suggest that the convergence rate can be improved if UAVs have a uniform successful sensing probability.
Based on this analysis, we formulate a training time minimization problem by jointly optimizing UAV deployment, integrated sensing, computation, and communication (ISCC) resources under a desirable optimality gap constraint.
To solve this challenging mixed-integer non-convex problem, we apply the alternating optimization technique, and propose the bandwidth, batch size, and position optimization (BBPO) scheme to optimize these three decision variables alternately.
Simulation results demonstrate that our BBPO scheme outperforms other baseline schemes regarding convergence rate and testing accuracy.
\end{abstract}

\begin{IEEEkeywords}
Federated edge learning, UAV deployment design, sensing-computation-communication resource allocation, integrated sensing and communication.
\end{IEEEkeywords}

\section{Introduction} \label{sec:introduction}

 \subsection{Motivations}
Advances in artificial intelligence (AI) and the Internet of Things (IoT) have shifted wireless networks from 5G's connected things to 6G's connected intelligence, with more stringent requirements for reliability and latency \cite{Roadmap6G, Survey6G, SurveyUAV, Chanelgain}.
In the pursuit of demanding requirements, federated edge learning (FEEL) has emerged as a popular distributed framework that trains global machine learning (ML) models at the edge of wireless networks \cite{FEEL, wei2023, chen2020joint, zhu2023pushing}, which helps preserve data privacy and avoid long delays caused by data transmission. 
Owing to the UAVs' mobility, UAV-mounted edge devices can facilitate data collection in FEEL.
A typical UAV-assisted FEEL iteration includes \emph{sensing, computation, and communication}, particularly for FEEL \cite{ISCC_liu}.
Specifically, starting with a common initial model broadcast from the edge server, the UAV-mounted edge devices collect/\emph{sense} data from the target and use their local sensing data to \emph{compute} the model updates. Then, they \emph{upload} their model updates to the edge server for further aggregation to yield an updated global model.

To fully exploit the benefits of UAV-assisted FEEL, overcoming several challenges associated with UAV deployment and limited onboard resources is crucial.
\emph{Firstly, UAV deployment affects the sensing data quality.} 
Due to practical constraints in sensing range and accuracy, the sensing data quality heavily relies on the sensing channel between the UAV and the target, which is controlled by the UAV position. This in turn affects the accuracy of model \cite{evidence1}. 
\emph{Secondly, UAV deployment affects the communication efficiency.}
The transmission rate of local model upload depends on the communication channel between the UAV and the server. As the channel model described in \cite{LOS_prob}, shorter relative distances and larger elevation angles favour line-of-sight (LOS) communication links. However, this improved communication efficiency leads to a decrease in sensing quality. 
Specifically, the shorter the UAV-server distance, the longer the UAV-target distance, resulting in a complex sensing environment, such as dense obstacles, which hinder the ability of UAVs to sense the target successfully.
As a result, the UAV deployment results in a tradeoff between sensing quality and communication efficiency, ultimately impacting the FEEL performance.
\emph{Thirdly, the limited onboard energy of UAVs imposes strict constraints on the duration of training.}
In order to minimize the training time, it is essential to allocate resources efficiently on both the UAV and server sides.
Therefore, we focus on the UAV deployment design and resource allocation in UAV-assisted FEEL to minimize the training time.

The preliminary studies \cite{UAV_FL_swarm, UAV_FL_energy,UAV_FL_image,UAV_FL_AQ,UAV_FL_blockchain} have mainly focused on implementing classical federated learning (FL) architectures using UAVs. Typically, UAVs are deployed as FL clients with stored datasets, or aerial base stations as model aggregators.
In recent research \cite{UAV_FL_relay, UAV_FL_TP,UAV_FL_placement, UAV_FL_hierarchical, UAV_FL_aggregator}, UAV deployment has been integrated into the FL system, enabling UAVs to serve as relays for efficient data collection from other devices, or as aggregators with an adaptive position to enhance the FL performance.
However, \emph{these existing studies rarely account for the data sensing process}, and assumed that the training data is either already available on the UAVs or can be directly acquired from other devices.
In addition, \emph{the sensing, computing, and communication processes are highly coupled in FEEL, competing for resources.}
While the most relevant study \cite{ISCC_liu} investigated the resource management of integrated sensing, computation, and communication (ISCC), it did not include UAV deployment in their problem formulation.
Therefore, we aim to address an open problem of \emph{optimizing the UAV deployment and ISCC resource allocation to enhance the training performance in UAV-assisted FEEL systems}.

 \subsection{Contributions}
In this paper, we propose a novel approach for human motion recognition using a UAV-assisted FEEL system.
The system comprises multiple UAVs and one edge server. Each UAV is considered an integrated sensing and communication (ISAC) device due to the size, weight, and power constraints \cite{ISAC_implement, ISAC_implement_motion}. 
In each training round, the ISAC UAVs sense the target and get their local datasets through wireless sensing. 
However, due to complex sensing environments such as trees, buildings, and other obstacles, UAVs may fail to sense targets successfully. Therefore, we employ a probabilistic sensing model, where the relative positions of UAVs and targets determine the successful sensing probability.
Only the UAVs that successfully sense the target can update their local models and upload them to the server for model aggregation.
To minimize the total training time, we optimize the UAV deployment and ISCC resource allocation, including bandwidth and batch size design, building upon a partial participation FEEL framework.
We aim to address three main challenges: 1) How does UAV deployment affect the sensing quality? 2) How does the probabilistic sensing model affects FEEL convergence? 3) How can we optimize the UAV deployment and ISCC resource allocation to improve the training performance?

\emph{To the best of our knowledge, this is the first work to investigate the impact of UAV deployment and ISCC resource allocation on the training performance of a UAV-assisted FEEL system.}

We summarize the key results and contributions as follows:

\begin{itemize}
\item \textbf{Impact of UAV deployment on sensing quality:}
We provide a theoretical analysis of the sensing process for human motion recognition and verify it experimentally.
Our results demonstrate that the sensing quality will reach \emph{a very good level} when the sensing elevation angle is larger than a given threshold.
This finding suggests that it is sufficient to sense targets at the aforementioned elevation angle threshold so that each UAV can generate high-quality data samples.
This addresses the first challenge above.
\item \textbf{FEEL convergence analysis under probabilistic sensing model:}
We first analyze the effect of the successful sensing probability on the convergence of FEEL by deriving the upper bound of the training loss.
Our derived theoretical results show that non-uniform successful sensing probabilities amplify the negative effects caused by data heterogeneity. 
However, these negative effects can be mitigated if UAVs have uniform successful sensing probabilities.
This addresses the second challenge above.
\item \textbf{Optimized UAV deployment and ISCC resource allocation:}
Based on solutions to the two challenges above, we formulate a joint UAV deployment and ISCC resource allocation problem to minimize the total training time under per round latency and training optimality gap constraints.
To solve this challenging mixed-integer non-convex problem, we adopt the alternating optimization technique and split it into three subproblems, which optimize the bandwidth, batch size, and UAV position, respectively.
Our proposed scheme can efficiently compute suboptimal solutions.
This addresses the third challenge above.
\item \textbf{Performance evaluation:}
We perform extensive simulations of a specific wireless sensing task (i.e., human motion recognition) on a high-fidelity wireless sensing simulator \cite{sense_platform}. Our bandwidth, batch size, and position optimization (BBPO) scheme achieves the best convergence rate and testing accuracy performance compared to other baseline schemes.

\end{itemize}

The rest of this paper is organized as follows. Section \ref{sec:related_works} discusses the related works. Section \ref{sec:system_model} describes the system model. Section \ref{sec:sense_analysis} presents the theoretical analysis of the sensing process. Section \ref{sec:problem_formulation} establishes the problem formulation. Section \ref{sec:train_time_minimize} develops the UAV deployment design and ISCC resource management. Section \ref{sec:simulation} provides the simulation results, and Section \ref{sec:conclusion} concludes this paper.

 \section{Related Works} \label{sec:related_works}
Initial studies on UAV-assisted FEEL focused primarily on implementing classical FL architectures using UAVs \cite{UAV_FL_swarm, UAV_FL_energy,UAV_FL_image,UAV_FL_AQ,UAV_FL_blockchain}.
Typically, the UAVs are deployed as FL clients with stored datasets or aerial base stations as model aggregators.
%
%
To fully reap the benefits of the UAV's mobility, recent works \cite{UAV_FL_relay, UAV_FL_TP,UAV_FL_placement, UAV_FL_hierarchical, UAV_FL_aggregator} have integrated UAV deployment into the UAV-assisted FEEL system, enabling UAVs to efficiently collect data from other devices, or act as an aggregator with adaptive position to enhance FL performance.
For example, Ng \emph{et al.} \cite{UAV_FL_relay} proposed using UAVs as wireless relays to facilitate communications between the Internet of Vehicles (IoV) components and the FL aggregator, thus improving the accuracy of the FL. Similarly, Lim \emph{et al.} \cite{UAV_FL_TP} proposed an FL-based sensing and collaborative learning approach for UAV-enabled IoVs, where UAVs collect data from subregions and train ML models for IoVs.
In \cite{UAV_FL_placement}, the authors proposed an asynchronous advantage actor-critic-based joint device selection, UAV placement, and resource management algorithm to enhance federated convergence rate and accuracy.
Additionally, the work in \cite{UAV_FL_hierarchical} introduced a new online FL scheme to improve the performance of personalized local models by jointly optimizing the CPU frequency and trajectories of the UAVs.
The work in \cite{UAV_FL_aggregator} proposed a joint algorithm for UAV placement, power control, bandwidth allocation, and computing resources, to minimize the energy consumption of the UAV aggregator and users.

\emph{Overall, these studies \cite{UAV_FL_swarm, UAV_FL_energy,UAV_FL_image,UAV_FL_AQ,UAV_FL_blockchain, UAV_FL_relay, UAV_FL_TP,UAV_FL_placement, UAV_FL_hierarchical, UAV_FL_aggregator} assumed that training data is readily available onboard the UAVs or can be obtained from other devices, regardless of the data sensing process.} However, this assumption can significantly impact training performance since sensing, computation, and communication in UAV-assisted FEEL are \emph{interdependent} and compete for resources.
The most related study \cite{ISCC_liu} focused on integrated sensing, computation, and communication (ISCC) resource management to optimize training performance.
However, they did not consider the UAV-mounted edge devices, which could significantly affect sensing and communication performance.
Therefore, we aim to address an open problem of optimizing ISCC resource allocation and UAV deployment to enhance the training performance in a UAV-assisted FEEL system.

\section{System Model} \label{sec:system_model}

As shown in Fig. \ref{fig:system_model}, we consider a UAV-assisted FEEL system consisting of $K$ UAVs and a server. The set of UAVs is represented as $\mathcal K=\{1,2,\dots, K\}$.
Each UAV is equipped with a single-antenna ISAC transceiver that can switch between the sensing and communication modes as needed in a time-division manner\footnote{A practical implementation of such an ISAC device via a software-defined radio platform has been demonstrated in \cite{ISAC_implement}.}.
Specifically, in the sensing mode, a dedicated radar waveform frequency-modulated continuous-wave (FMCW) \cite{ISAC_implement_motion} is transmitted. Then, sensing data containing the motion information of the human target can be obtained on board the UAVs by processing the received radar echo signals. This mode is used for UAVs to collect training data.
On the other hand, in the communication mode, a constant frequency carrier modulated by communication data is transmitted. This mode is used for information exchange between UAVs and the edge server.
%
Our UAV-assisted FEEL system aims to train an ML model for specific aerial sensing applications, e.g., target recognition, using the data wirelessly sensed by the UAVs.

\begin{figure}[t]
\begin{center}
\includegraphics[width=8.7cm, trim = 0cm 0cm 0cm 0cm, clip = true]{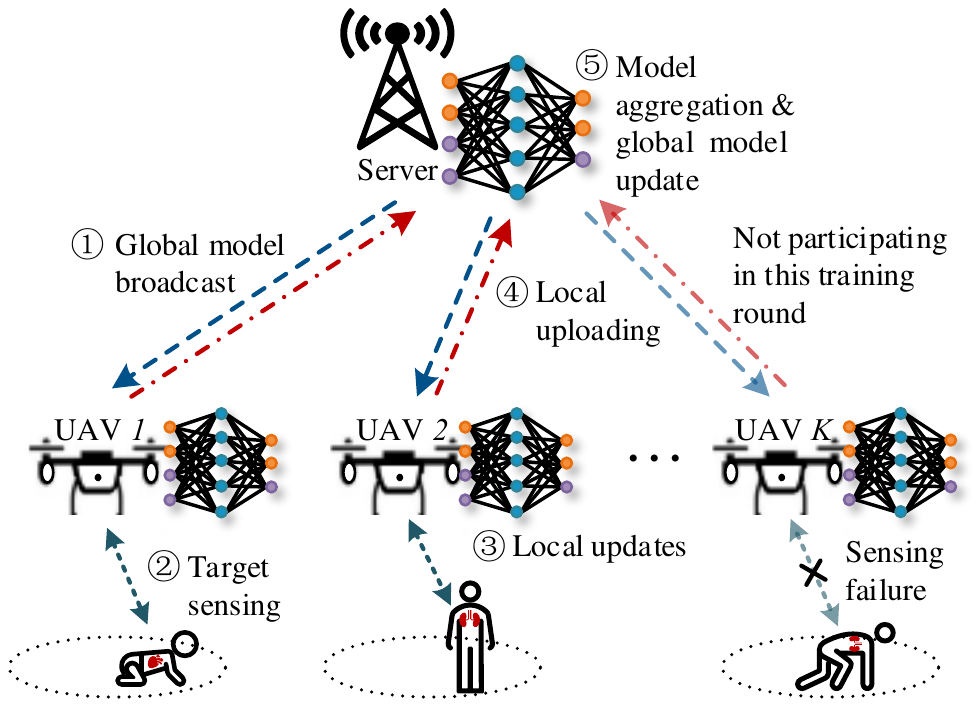}
\caption{\footnotesize{UAV-assisted FEEL system with integrated sensing, communication, and computation.}}
\label{fig:system_model}
\end{center}
\end{figure}

\subsection{Learning Model} \label{sec:learning_model}
The training process is to minimize a global loss function in a distributed manner. In particular, we define the global loss function as
\begin{equation}
F(\mathbf w)=\frac {1} {K} \sum\limits_{k\in\mathcal K}\mathbb E_{\varepsilon\sim\mc P_k}\left[f_k(\mathbf w;\varepsilon)\right],
\end{equation}
where $\mathbf w$ represents the global model, jointly trained on all UAVs and coordinated with the edge server. The function $f_k(\mathbf w;\varepsilon)$ is the local loss function for UAV $k$, and $\varepsilon$ is a random seed with distribution $\mc P_k$, of which its realization represents a batch of samples. To facilitate the subsequent analysis, we define $F_k(\mathbf w)=\mathbb E_{\varepsilon\sim\mc P_k}\left[f_k(\mathbf w;\varepsilon)\right]$.
The training process is iterated in multiple communication rounds. The system repeats the following five steps in each round $n$ until the global model converges (see Fig. \ref{fig:system_model}).

\begin{enumerate}
\item \textbf{Global model broadcasting:}
The server broadcasts the current global model $\mathbf w^{(n)}$ to each UAV via the wireless broadcast channel.
\item \textbf{Target sensing:}
Each UAV switches to the sensing mode and transmits dedicated FMCW signals for target sensing.
We assume that the UAV can only sense the target successfully if it has a LOS sensing link with the target and no obstacles in the link.
Due to the unpredictable sensing environment, we utilize the LOS probabilistic model\footnote{The formulation and the analysis can be readily applied to other probabilistic sensing models described in \cite{radar_prob1, radar_prob2}.} for sensing.
We define the 3-dimensional coordinates of UAV $k$ as $\bs u_k=(x_{u,k},y_{u,k},z_{u,k})$, where $z_{u,k}$ represents the flying altitude.
The 3-dimensional coordinates of the sensed target is $\bs v_k=(x_{v,k},y_{v,k},z_{v,k})$, where $z_{v,k}$ is the altitude of the corresponding target.
The successful sensing probability of UAV $k$ is \cite{LOS_prob}
\begin{equation} \label{equ:los_s}
    q_{s,k}(\bs u_k)=\frac 1 {1+\psi \exp(-\zeta[\theta_{s,k}(\bs u_k)-\psi])},
\end{equation}
where $\psi$ and $\zeta$ are constant values determined by the type of environment, $\theta_{s,k}(\bs u_k)$ is the sensing elevation angle between UAV $k$ at position $\bs {u}_k$ and its sensed target.
More specifically, $\theta_{s,k}(\bs u_k)=\frac {180^{\circ}} \pi \times \sin^{-1}\left |\frac {z_{u,k}-z_{v,k}} {d_{s,k}}\right|$, where $d_{s,k}=\|\bs u_{k}-\bs v_k\|$ is the Euclidean distance between UAV $k$ and the sensed target.
Once UAV $k$ successfully senses the target in round $n$, it obtains a batch of data samples with size $\delta_k$.
The batch size $\delta_k$ can vary adaptively across different UAVs but remain unchanged across all training rounds.

\item \textbf{Local gradient updating:}
Each UAV that successfully senses the target updates its local gradient by running one step of the stochastic gradient descent (SGD) from $\mathbf w^{(n)}$, i.e.,
\begin{equation}
\mathbf g_k^{(n)}=\frac 1 {\delta_k} \sum\limits_{\xi\in\mathcal D_k^{(n)}} \nabla f_k(\mathbf w^{(n)};\xi).
\end{equation}
For UAVs that do not successfully sense the targets, their local gradients will not be updated, i.e., $\mathbf g_k^{(n)}=\mathbf g_k^{(n-1)}$.

\item \textbf{Local uploading:}
We assume that only successfully sensed UAVs upload their local gradients to the server via the uplink
wireless channel (see Section \ref{sec:sense_analysis} for more details on the sensing model). Therefore, we consider the partial participating FEEL scenario.

\item \textbf{Global model aggregation and updating:}
The server aggregates local gradients and updates the global model as
\begin{equation} \label{equ:w}
\mathbf  w^{(n+1)}=\mathbf  w^{(n)} - \eta\frac {\sum\limits_{k\in\mc K} \mathbf 1_{k}^{(n)}(\bs u_k)\mathbf {g}_k^{(n)}} {\sum\limits_{k\in\mc K} \mathbf 1_{k}^{(n)}(\bs u_k)},
\end{equation}
\end{enumerate}
where $\eta$ is the learning rate.
The indicator function $\mathbf 1_{k}^{(n)}(\bs u_k)$ is
\begin{equation} \label{equ:q_indicator}
\mathbf 1_{k}^{(n)}(\bs u_k)=
\begin{cases}
1, & \text{if UAV $k$ senses its target successfully in round $n$}\ \text{w.p.}\ q_{s,k}(\bs u_k),\\
0, & \text{if UAV $k$ senses its target unsuccessfully in round $n$}\ \text{w.p.}\ 1-q_{s,k}(\bs u_k).
\end{cases}
\end{equation}
%

\subsection{UAV-server Communication Model} \label{sec:communication}
%
Without loss of generality, we define the 3-dimensional coordinates of the server as $\bs s=(x_s,y_s,z_s)$, where $z_s$ is the altitude of the server.
Next,
we analyze the channel model between UAV $k$ and the server, which can be LOS or non-LOS (NLOS).
The LOS probability in each round $n$ is \cite{LOS_prob}
\begin{equation} \label{equ:LOS probability}
    q_{c,k}(\bs u_k)=\frac 1 {1+\psi \exp(-\zeta[\theta_{c,k}(\bs u_k)-\psi])},
\end{equation}
where $\theta_{c,k}(\bs u_k)$ is the communication elevation angle. More specifically, $\theta_{c,k}(\bs u_k)=\frac {180^{\circ}} \pi \times \sin^{-1}\left|\frac {z_{u,k}-z_s} {d_{c,k}}\right|$, where $d_{c,k}=\|\bs u_{k}-\bs s\|$ is the Euclidean distance between UAV $k$ and the server.
For ease of analysis, we consider that the UAVs are flying at a constant altitude, i.e., $z_{u,k}=H,\ \forall k\in\mc K$.
Moreover, the NLOS probability is $1-q_{c,k}(\bs u_k)$.
%
The average channel gain between UAV $k$ and the server is \cite{R1, Chanelgain} 
\begin{equation} \label{equ:ATG channel gain}
    h_{k}(\bs u_k)= \frac {(K_0d_{c,k})^{-\alpha}} {\eta_1 q_{c,k}(\bs u_k)+\eta_2 (1-q_{c,k}(\bs u_k))},
\end{equation}
where $K_0=\frac {4\pi f_c} c$, $f_c$ is the carrier frequency, $c$ is the speed of light, and $\alpha$ is the path loss exponent of the link between the UAV and the server. Also, $\eta_1$ and $\eta_2$ ($\eta_2 > \eta_1>1$) are the excessive path loss coefficients in LOS and NLOS cases, respectively, which models the shadowing effect and reflection effect \cite{LOS_prob}.

We consider that each UAV is assigned a dedicated sub-channel of bandwidth $B_k$ for uploading.
Suppose that UAV $k$ uploads the gradient with a power of $p_{c,k}>0$.
The received rate of the server from UAV $k$ is
\begin{equation} \label{equ:UAV-server}
\begin{array}{rll}
    \!\!\!\!\!\!&r_{k}(\bs u_k,B_{k})
    \displaystyle =B_{k}\log\left(1+\frac {p_{c,k} h_{k}(\bs u_k)}{B_kN_0}\right),
\end{array}
\end{equation}
where the bandwidth allocation is limited by the total bandwidth $B_c$ such that $\sum\limits_{k \in \mc{K}} B_k = B_c$, and $N_0$ is the noise power spectral density.

\subsection{UAV-target Sensing Model}\label{sec:sensing}
%
The dedicated FMCW signal consisting of multiple up-chirps is used for UAVs' sensing \cite{ISAC_implement_motion}.
The sensing signal of UAV $k$ at time $t$ is defined as $x_k(t)$.
We use the primitive-based method \cite{Sens_doppler1} to model the scattering from the whole human body to UAVs. The human contains $L$ body primitives.
The scattering along a direct/indirect reflection path is approximated using the superposition of the returns from $L$ body primitives.
The signal received by UAV $k$ at time $t$ is\footnote{
UAVs are performing human motion recognition in different places and far apart. Therefore, UAVs can sense their targets through the same sensing bandwidth without interfering with each other. } 
\begin{equation} \label{equ:sense_signal}
\begin{array}{rll}
    \!\!\!\!\!\! y_{k}(t)
    \displaystyle =\sqrt {p_{s,k}(t)}\times\frac {A_0} {\sqrt{4\pi}}\sum\limits_{l=1}^L\frac{\sqrt{G_{k,l}(t)}}{d_{k,l}^2(t)}\exp\left(-j\frac{4\pi f_c}{c}d_{k,l}(t)\right)x_{k,l}\left(t-2\frac{d_{k,l}(t)}{c}\right) + n_k(t),
\end{array}
\end{equation}
where $p_{s,k}(t)$ is the sensing transmission power of UAV $k$, $A_0$ is the gain of the antenna, $G_{k,l}(t)$ is the complex amplitude proportional to the radar cross section (RCS) of the $l$-th primitive, $d_{k,l}(t)$ is the distance from the $l$-th primitive to UAV $k$, and $n_k(t)$ is the signal due to ground clutter and noise.
Note that we only consider the direct reflection path in \eqref{equ:sense_signal}, since UAVs in aerial space are less susceptible to indirect reflection paths.

\subsection{Training Time}
%
Based on the learning procedure in Section \ref{sec:learning_model}, the latency for each UAV $k$ in training round $n$ consists of the following three parts\footnote{We ignore the global download time of step 1) in Section \ref{sec:learning_model} because the server can use the entire frequency band to broadcast the global model to all UAVs. Usually, the server has a more considerable transmit power.}:

\begin{itemize}
\item \emph{Sensing time:} We consider that each UAV takes time $T_0$ to generate a sample (see Section \ref{sec:sense_analysis} for details). The number of samples generated by UAV $k$ is $\delta_{k}$. Therefore, the sensing time of UAV $k$ in round $n$ is
\begin{equation} \label{equ:sensing time}
T_{s,k}^{(n)}(\delta_k)=T_0\cdot \delta_k.
\end{equation}
%
\item \emph{Local computation time:} The indicator function $\mathbf 1_{k}^{(n)}(\bs u_k)$ defined in \eqref{equ:q_indicator} indicates whether UAV $k$ successfully senses the target. Since UAVs partly participate in FEEL, the local calculation time of UAV $k$ in round $n$ is
\begin{equation} \label{equ:local updating time}
T_{cp,k}^{(n)}(\delta_k,\bs u_k)=\mathbf 1_{k}^{(n)}(\bs u_k) \cdot\frac {\delta_k \xi}{f_{\text{cpu}}},
\end{equation}
where $\xi$ is the CPU cycles required to execute one sample in the local gradient computing, and $f_{\text{cpu}}$ is each UAV's CPU frequency (cycles/s).

\item \emph{Local upload time:} Each UAV that successfully senses the target needs to upload its gradient to the server. We consider the data size of the gradient to be a constant, defined as $D_0$.
The local upload time of UAV $k$ in round $n$ is
\begin{equation} \label{equ:local uploading time}
T_{cm,k}^{(n)}(\bs u_k,B_k)= \mathbf 1_{k}^{(n)}(\bs u_k) \cdot\frac {D_0}{r_k(\bs u_k,B_k)}.
\end{equation}
\end{itemize}

We consider the synchronous FL on the server side, i.e., aggregation happens until local gradients from all participating UAVs are received. Then, the latency for round $n$ is\footnote{For ease of illustration, we denote $\{\delta_k, \forall k\in\mc K\}$, $\{\bs u_k, \forall k\in\mc K\}$, and $\{B_k, \forall k\in\mc K\}$ as $\{\delta_k\}$, $\{\bs u_k\}$, and $\{B_k\}$, respectively.}
\begin{equation} \label{equ:round latency}
T^{(n)}\left(\{\delta_k\}, \{\bs u_k\}, \{B_k\}\right)= \max\limits_{k\in\mc K}\left\{T_{s,k}^{(n)}(\delta_k)+T_{cp,k}^{(n)}(\delta_k, \bs u_k)+T_{cm,k}^{(n)}(\bs u_k,B_k)\right\}.
\end{equation}
%

\section{Sensing Analysis} \label{sec:sense_analysis}

In this section, we analyze the effect of UAV deployment on the quality of sensing samples.
Micro-Doppler signature is a characteristic of human motion. 

We aim to verify our hypothesis that there exists a threshold value of the sensing evaluation angle that can guarantee a satisfactory quality of the data samples.
A common technique for micro-Doppler analysis is time-frequency representation, such as spectrograms\cite{ISAC_implement_motion}. Therefore, the received signal in \eqref{equ:sense_signal} requires further preprocessing to obtain the micro-Doppler signature.
Specifically, each spectrogram is generated from the received signal over a time period $T_0=MT_p$, where $M$ is the number of chirps and $T_p$ is the duration of each chirp.
Since each UAV $k$ can control $\delta_k$ collected samples (i.e. spectrograms), as defined in Section \ref{sec:learning_model}, UAV $k$ needs to sense the target for a continuous duration of $\delta_kT_0$.
The received raw signal in \eqref{equ:sense_signal} first needs to be sampled with rate $f_s$.
After that, the signal in the $\rho$-th sensing duration can be reconstructed as a 2D sensing data matrix $\mathbf Y_{k}(\rho)\in \mathbb C^{f_sT_p\times M}$ \cite{ISCC_liu}, i.e.,
\begin{equation} \label{equ:sense_signal_matrix}
\begin{array}{rll}
    \!\!\!\!\!\!&\mathbf Y_{k}(\rho)
    \displaystyle =\sqrt {p_{s,k}(\rho)}\times\frac {A_0} {\sqrt{4\pi}}\sum\limits_{l=1}^L\frac{\sqrt{G_{k,l}(\rho)}}{d_{k,l}^2(\rho)}\exp\left(-j\frac{4\pi f_c}{c}d_{k,l}(\rho)\right)\mathbf X_{k,l}(\rho) + \mathbf N_k(\rho).
\end{array}
\end{equation}
Sensing power $p_{s,k}(\rho)$ and distance $d_{k,l}(\rho)$ will remain constant for each $\rho$-th duration, where $\rho\in [ 1, 2, \dots, \delta_k ]$.
The element of row $r$ and column $m$ of matrix $\mathbf Y_{k}(\rho)$ is $[\mathbf Y_{k}(\rho)]_{r,m}=y_k\left((\rho-1)T_0\\ +(m-1)T_p+\frac{r}{f_s}\right)$, $[\mathbf X_{k,l}(\rho)]_{r,m}=x_{k,l}\left((\rho-1)T_0+(m-1)T_p+\frac{r}{f_s}-2\frac{d_{k,l}(\rho)}{c}\right)$, and $[\mathbf N_{k}(\rho)]_{r,m}=n_k\left((\rho-1)T_0+(m-1)T_p+\frac{r}{f_s}\right)$, where $r\in [1,2, \dots, f_sT_p]$ represents the sampled signal index, and $m\in [1, 2, \dots, M]$ represents the chirp index.

Next, we adapt the short-time Fourier transform (STFT) with a sliding window function $o[w]$ of length $W$ to generate the range-Doppler-time (RDT) cube in order to obtain time-frequency features of $\mathbf Y_{k}(\rho)$. The RDT cube of $\mathbf Y_{k}(\rho)$ after STFT is given by \cite{ISAC_implement_motion, ISCC_liu}
\begin{equation} \label{equ:sense_signal_matrix_FT}
\begin{array}{lll}
    \!\!\!\! \widetilde{\mathbf Y}_{k}(\rho,f)
    \displaystyle =\sqrt {p_{s,k}(\rho)}\times\frac {A_0} {\sqrt{4\pi}}\sum\limits_{l=1}^L\frac{\sqrt{G_{k,l}(\rho)}}{d_{k,l}^2(\rho)}\exp\left(-j\frac{4\pi f_c}{c}d_{k,l}(\rho)\right)\widetilde{\mathbf X}_{k,l}(\rho,f) + \widetilde{\mathbf N}_k(\rho,f).
\end{array}
\end{equation}
The element of $\widetilde{\mathbf X}_{k,l}(\rho,f)$ is
\begin{equation} \label{equ:sense_signal_matrix_FT_par1}
\begin{array}{rll}
    \!\!\!\!\!\! \left[\!\widetilde{\mathbf X}_{k,l}(\rho,f)\!\right]_{r,m}
    \!\!\!=\!\!\sum\limits_{w=0}^W\!x_{k,l}\!\left(\!(\rho\!-\!1)T_0\!+\!\left(m(W\!-\!Q)\!-\!w\!-\!1\right)T_p\!+\!\frac{r}{f_s}\!-\!2\frac{d_{k,l}(\rho)}{c}\!\right)\!\exp(\!-\!j\frac{2\pi fw}{W})o(w),
\end{array}
\end{equation}
where $Q$ is the number of overlapping points, $f\in[0,W-1]$ and $m\in\left[1,\frac{M-Q}{W-Q}\right]$ are the frequency and temporal shift index, respectively.
We non-coherently integrate $\widetilde{\mathbf X}_{k,l}(\rho,f)$ within all the range bins, and obtain the integrated STFT as $\left[\overline{\mathbf X}_{k,l}(\rho,f)\right]_m=\sum\limits_{r=1}^{fsT_p}\left[\widetilde{\mathbf X}_{k,l}(\rho,f)\right]_{r,m}$.
Moreover, the element definition of $\widetilde{\mathbf N}_{k}(\rho,f)$ is similar to $\widetilde{\mathbf X}_{k,l}(\rho,f)$. 
Then, we have
\begin{equation} \label{equ:sense_signal_matrix_FT_int}
\begin{array}{rll}
    \!\!\!\! \overline{\mathbf Y}_{k}(\rho,f)
    \displaystyle =\sqrt {p_{s,k}(\rho)}\times\frac {A_0} {\sqrt{4\pi}}\sum\limits_{l=1}^L\underbrace{\frac{\sqrt{G_{k,l}(\rho)}}{d_{k,l}^2(\rho)}\exp\left(-j\frac{4\pi f_c}{c}d_{k,l}(\rho)\right)\overline{\mathbf X}_{k,l}(\rho,f)}_{\text{Useful information}}
     + \underbrace{\overline{\mathbf N}_k(\rho,f)}_{\text{Interference}},
\end{array}
\end{equation}
where the clutter and noise related term $\overline{\mathbf N}_k(\rho,f)$ represents the interference, and the rest represents useful information in the spectrogram.

\begin{figure*}[t]
\hspace{-0.5cm}
\centering
\setlength{\belowcaptionskip}{-0.3cm}
\setlength{\abovecaptionskip}{-0.1cm}
\begin{minipage}[t]{0.47\linewidth}
\includegraphics[width=7.5cm, trim = 0cm 0cm 0cm 0cm, clip = true]{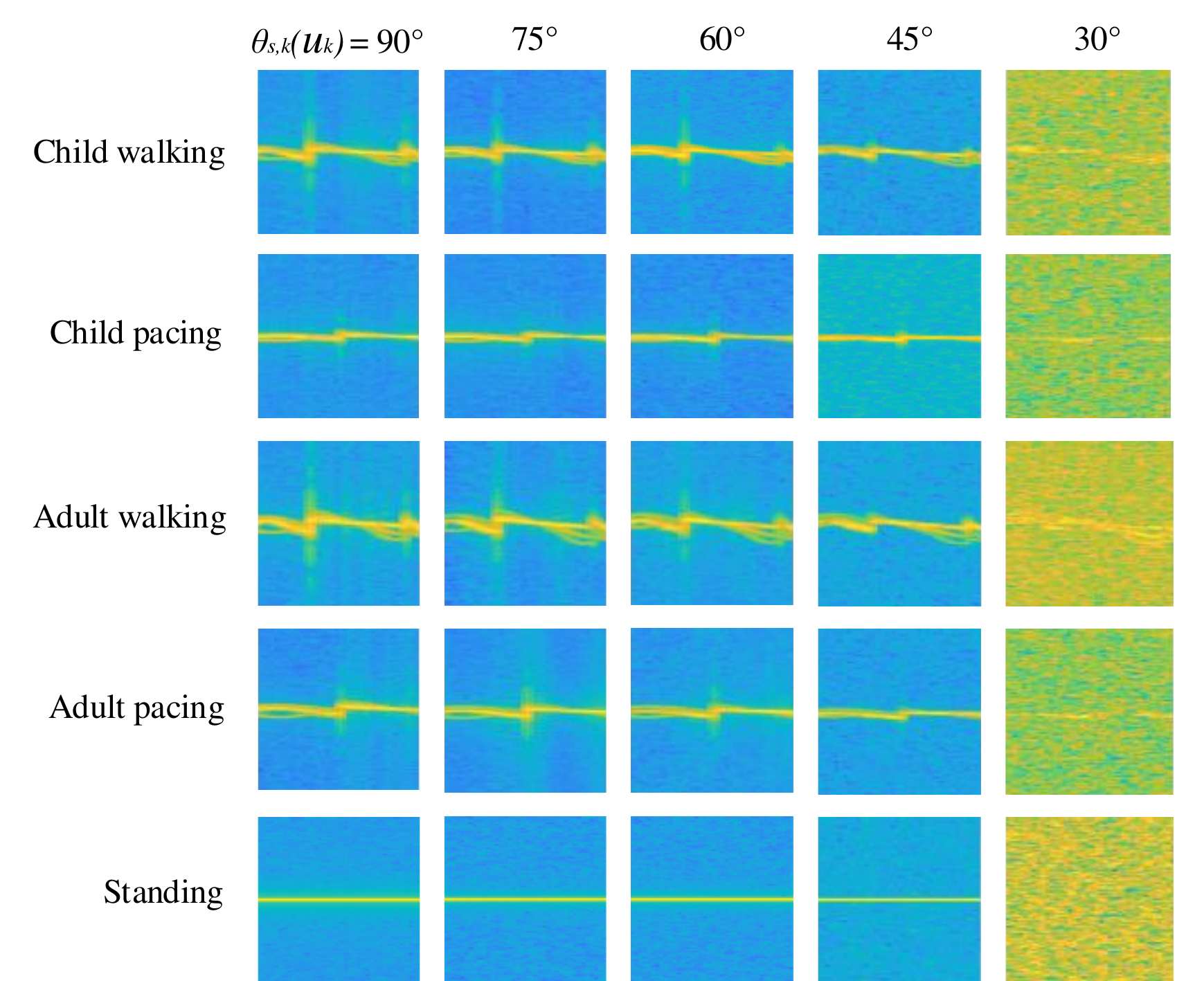}
\caption{\footnotesize{The spectrograms of five human motions at different sensing elevation angles $\theta_{s,k}(\bs u_k)$.}}
\label{fig:quality}
\end{minipage}
\quad
\begin{minipage}[t]{0.47\linewidth}
\includegraphics[width=6.5cm, trim = 0cm 0cm 0cm 0cm, clip = true]{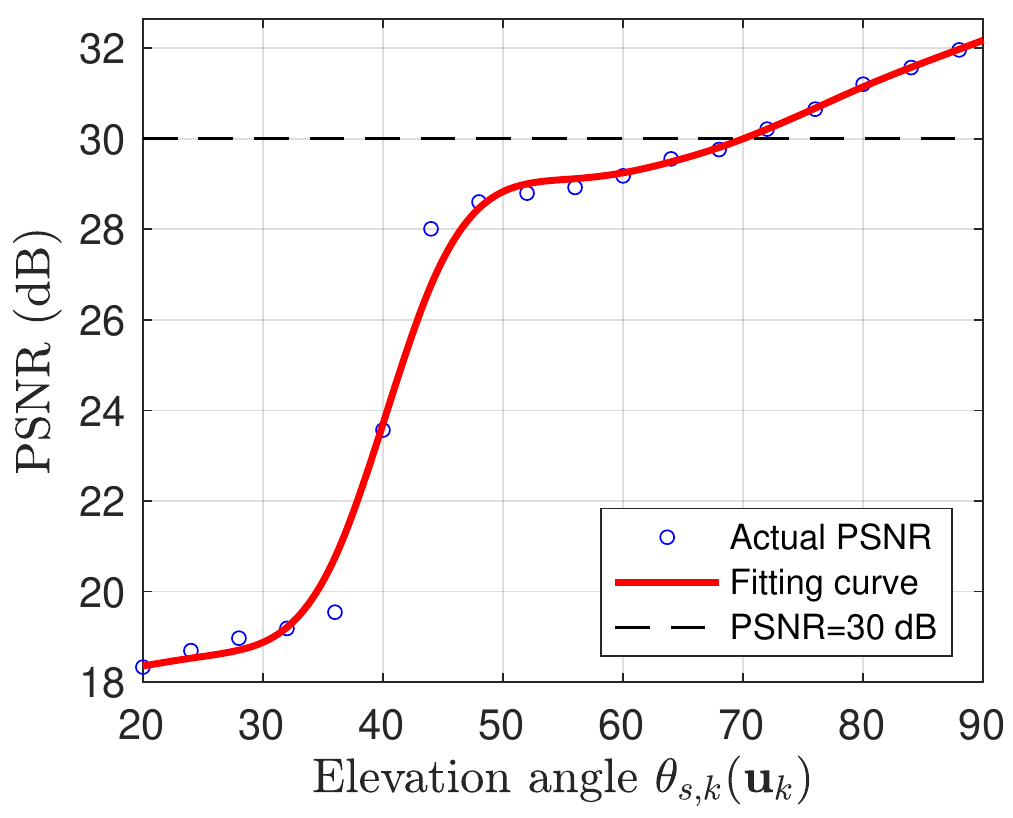}
\caption{\footnotesize{The quality of spectrogram PSNR versus sensing elevation angle $\theta_{s,k}(\bs u_k)$.}}
\label{fig:PSNR}
\end{minipage}
\end{figure*}

\begin{remark} \label{remark:quality}
\textbf{Impact of UAV position $\{\bs u_k\}$ on spectrogram quality.}
\emph{From \eqref{equ:sense_signal_matrix_FT_int}, we can see that the quality of spectrograms improves as we reduce $d_{k,l}(\rho)$, which is the distance between UAV $k$ and the $l$-th primitive of the human body.
Since each UAV $k$ flies at a constant altitude, distance $d_{k,l}(\rho)$ is minimized when the UAV hovers over the target, which leads to the best quality of the spectrograms.
The experimental results also validate this in Fig. \ref{fig:quality} and Fig. \ref{fig:PSNR}.
For ease of illustration, we show the relationship between the spectrogram and sensing elevation angle $\theta_{s,k}(\bs u_k)\in[0^\circ, 90^\circ]$. Note that higher $\theta_{s,k}(\bs u_k)$ leads to smaller $d_{k,l}(\rho)$.
We apply the wireless sensing simulator in \cite{sense_platform} to simulate various human motions and generate five human motion spectrograms as shown in Fig. \ref{fig:quality}. We can observe that when $\theta_{s,k}(\bs u_k)=30^\circ$, the quality of the spectrogram becomes the worst.
Moreover, we use the peak-signal-to-noise ratio (PSNR) index\footnote{
PSNR is a metric measuring the perceptual difference between two similar images, widely used in image compression and computer vision. It is a complete reference metric that requires two images, i.e., a reference image and a processed image.}
to visualize the quality of the spectrogram \cite{PSNR}.
Fig. \ref{fig:PSNR} shows the average PSNR versus $\theta_{s,k}(\bs u_k)$, and a fitting curve that is very close to actual PSNR.
We can see that PSNR increases with $\theta_{s,k}(\bs u_k)$ and reaches the highest value at $\theta_{s,k}(\bs u_k)=90^\circ$.
Since the spectrograms have a \emph{very good level} of quality when PSNR$\geq 30$ dB \cite{PSNR_30dB}, we can set a threshold for the sensing evaluation angle, e.g. $\theta_{s,k}(\bs u_k)\geq \theta_0 $.
As a result, each UAV can generate data samples (spectrograms) of roughly satisfactory quality.}

\end{remark}

\section{Problem Formulation} \label{sec:problem_formulation}

In this section, we aim to speed up the training process under a specific optimality gap of the loss function.
In Sections \ref{sec:assumption} and \ref{sec:theorm_results}, we analyze the convergence of the UAV-assisted FEEL training process and derive an upper bound on the loss function.
In Section \ref{sec:problem}, we formulate the training time minimization problem through bandwidth allocation, batch size design, and UAV position design, which balances the sensing, computation, and communication during training.

\subsection{Assumptions} \label{sec:assumption}

We consider general smooth convex learning problems with the following commonly adopted assumptions \cite{SGD_1,SGD_2, quantize, Zhu_one_bit}.
\begin{assumption}\label{assump:smooth}
(Smoothness).
Each local function $F_k(\mathbf w)$ is Lipschitz continuous with Lipschitz constant $L$:
$\forall \mathbf w_i$ and $\mathbf w_j$, $F_k(\mathbf w_i)\leq F_k(\mathbf w_j)+(\mathbf w_i-\mathbf w_j)^\mathrm{T}\nabla F_k(\mathbf w_j)+\frac L 2 \Vert\mathbf w_i-\mathbf w_j\Vert^2,\ \forall k\in\mc K$.
\end{assumption}
\begin{assumption}\label{assump:convex}
(Strong convexity).
Each local function $F_k(\mathbf w)$ is strongly convex in that there exists a constant $\mu>0$ such that
$\forall \mathbf w_i$ and $\mathbf w_j$, $F_k(\mathbf w_i)\geq F_k(\mathbf w_j)+(\mathbf w_i-\mathbf w_j)^\mathrm{T}\nabla F_k(\mathbf w_j)+\frac {\mu} 2 \Vert\mathbf w_i-\mathbf w_j\Vert^2,\ \forall k\in\mc K$.
\end{assumption}
\begin{assumption}\label{assump:Fg}
(Unbiasedness and bounded variance of local gradients).
The mean and variance of stochastic gradient $\mathbf g^{(n)}_k$ of local loss function $F_k(\mathbf w)$, $ \forall k\in\mc K$, satisfy that
\begin{equation} \label{equ:unbiased_exp_var}
\begin{array}{rll}
&\mathbb E \left[\mathbf g_k^{(n)}\right]=\nabla F_k(\mathbf w^{(n)}),\\
&\mathbb E\left[\left\Vert \mathbf g_k^{(n)}-\nabla F_k(\mathbf w^{(n)})\right\Vert^2\right]
\leq\displaystyle\frac{\sigma_k^2}{\delta_k},\nonumber
\end{array}
\end{equation}
where $\delta_k$ is the batch size in calculating gradient $\mathbf g_k^{(n)}$.
\end{assumption}
\begin{assumption}\label{assump:FF}
(Bounded data variance).
\begin{equation} \label{equ:data_var}
\begin{array}{rll}
\mathbb E\left[\left\Vert\nabla F_k(\mathbf w^{(n)})-\nabla F(\mathbf w^{(n)})\right\Vert^2\right]
\leq\Lambda_k^2, \nonumber
\end{array}
\end{equation}
which measures the heterogeneity of local datasets. 
\end{assumption}
%

\subsection{Theoretical results} \label{sec:theorm_results}
Under Assumptions \ref{assump:smooth}-\ref{assump:FF}, the following theorem establishes the convergence rate.
\begin{theorem}\label{theorem:Fbound}
Consider the UAV-assisted FEEL system with a fixed learning rate $\eta$, satisfying
\begin{equation} \label{equ:learningratec}
\begin{array}{rll}
0<\eta<\displaystyle\frac 1 {4L}.
\end{array}
\end{equation}
The expected optimality gap of the loss function after $n$ rounds is upper bounded by
\begin{equation} \label{equ:Fbound}
\begin{array}{lll}
\!\!\!\!\Exp{F(\mathbf w^{(n)}) \!\!-\!\! F(\mathbf w_*)} \! \leq \! G\displaystyle\frac {1\!-\! \left(1\!-\!\mu\eta\left(1\!-\!4L\eta\right)\right)^n}{\mu\eta\left(1\!-\!4L\eta\right)}
\!+\!\left(1\!-\!\mu\eta\left(1\!-\!4L\eta\right)\right)^n\Exp{F(\mathbf w^0)\!-\!F(\mathbf w_*)},
\end{array}
\end{equation}
with
\begin{equation} \label{equ:G}
\begin{array}{lll}
\!\!\!\!\! G\!=\!\eta\!\sum\limits_{k\in\mc K}\!\alpha_k^2 \!\sum\limits_{k\in\mc K}\!\left(\displaystyle\frac{\sigma_k^2}{\delta_k}\!+\!\Lambda_k^2\right)
\!+\!\displaystyle L\eta^2\!\sum\limits_{k\in\mc K}\!\displaystyle\beta_k \!\left(\frac{\sigma_k^2}{\delta_k}\!+\!2\Lambda_k^2\right)
\!+\!2L\eta^2 \!\sum\limits_{k\in\mc K}\!\gamma_k \!\left(\left(q_{s,k}(\bs u_k)\!-\!\overline q_s\right)^2\!+\!\overline q_s^2\right)\!\Lambda_{k}^2,
\end{array}
\end{equation}
where $\mathbf w_*$ is the optimal model as $\mathbf w_*=\arg\min\limits_{\mathbf w}F(\mathbf w)$, each $\alpha_k, \beta_k, \gamma_k, \forall k\in\mc K$ is a function of $\left\{q_{s,k}(\bs u_k), \forall k\in\mc K\right\}$, as shown in equations \eqref{equ:g}, \eqref{equ:M11}, and \eqref{equ:M22final} in Appendix \ref{app:Fbound}, and $\overline q_s=\frac 1 K\sum\limits_{k\in\mc K}q_{s,k}(\bs u_k)$ is the average successful sensing probability of UAVs.
We represent the right hand side of \eqref{equ:Fbound} as $\Phi\left(\{\delta_{k}\},\{\bs u_k\}, n\right)$, and we have $\Exp{F(\mathbf w^{(n)})-F(\mathbf w_*)}\leq\Phi\left(\{\delta_{k}\},\{\bs u_k\}, n\right)$.
\end{theorem}
\begin{proof} \label{proof:Fbound}
See Appendix \ref{app:Fbound}.
\end{proof}

It can be seen from \eqref{equ:G} that when the UAVs have different successful sensing probabilities $q_{s,k}(\bs u_k)$, the negative effects caused by data heterogeneity $2L\eta^2 \!\sum\limits_{k\in\mc K}\!\gamma_k \!\left(\left(q_{s,k}(\bs u_k)\!-\!\overline q_s\right)^2\!+\!\overline q_s^2\right)\!\Lambda_{k}^2$ will be amplified.
However, these negative effects can be mitigated if UAVs have uniform successful sensing  probabilities $q_{s,k}(\bs u_k)=q_s, \forall k\in \mc K$. 
In this case, the training process can still converge to a proper stationary solution. 
Next, we derive the following corollary.

\begin{corollary}\label{corollary:Gsameq}
Under learning rate constraint \eqref{equ:learningratec}, when UAVs have the uniform successful sensing probability (i.e., $q_{s,k}(\bs u_k) = q_s, \forall k \in \mc{K}$), we can obtain an upper bound of $G$ in \eqref{equ:G} as
\begin{equation} \label{equ:Gsameq}
\begin{array}{rll}
G\leq\displaystyle\frac {\eta} {K}\sum\limits_{k\in\mc K}\left(\displaystyle\frac{\sigma_k^2}{\delta_k}+\Lambda_k^2\right)
+\frac{2L\eta^2}{K^2 \chi q_s}\sum\limits_{k\in\mc K}\left(\displaystyle\frac{\sigma_k^2}{\delta_k}+2\Lambda_k^2\right)
+\frac{4L\eta^2}{K q_s^{K-2}}\sum\limits_{k\in\mc K}\Lambda_{k}^2,
\end{array}
\end{equation}
where $\chi=1-(1-q_{s,\text{min}})^K$, and $q_{s,\text{min}}=\min\limits_{\bs u_k}q_{s,k}(\bs u_k)$ with constraint $\theta_{s,k}(\bs u_k)\geq\theta_0, \forall k \in \mc{K}$.

\end{corollary}


%
\begin{proof} \label{proof:corollary}
See Appendix \ref{app:Gsameq}.
\end{proof}
\begin{remark}\label{remark:NT_tradeoff}
%
From Theorem \ref{theorem:Fbound} and Corollary \ref{corollary:Gsameq}, we have two important insights:
\textbf{1) From \eqref{equ:Fbound} and \eqref{equ:Gsameq}, it can be seen that the convergence rate is affected by the batch size $\{\delta_k\}$ and the successful sensing probability $q_s$.}
Specifically, we  observe that a decrease in $\{\delta_k\}$ slows down the training convergence in \eqref{equ:Gsameq}, because fewer data samples are used for each training round.
Moreover, a smaller $q_s$ also deteriorates the training convergence in \eqref{equ:Gsameq}.
Since only UAVs that successfully sense the targets can participate in the training, a smaller $q_s$ represents the smaller successful sensing probability, and the smaller probability of the UAVs participating in the training.
Therefore, the left hand side of \eqref{equ:Fbound} takes more training rounds to converge.
\textbf{2) From \eqref{equ:Fbound}, the loss function eventually converges as $n\to\infty$, and $\Exp{F(\mathbf w^{(n)}) \!\!-\!\! F(\mathbf w_*)}$ will reach $\frac{G}{\mu\eta(1-4L\mu)}$ instead of diminishing to zero.}
This implies that partial participation of UAVs has a negative impact on the training process, and causes the loss function to converge to a biased solution.
Furthermore, the biased solution $\frac{G}{\mu\eta(1-4L\mu)}$ strongly depends on $q_s$ and $\{\delta_k\}$, since it contains $G$ restricted by \eqref{equ:Gsameq}.
\end{remark}

\subsection{Problem Formulation}\label{sec:problem}
As probabilistic sensing is inevitable in delay-constrained wireless sensing systems, our goal is to investigate its impact on the training time of UAV-assisted FEEL.
From \eqref{equ:los_s}, we know that the successful sensing probability is determined by UAV position.
Therefore, we formulate a resource allocation problem to minimize the total training time by optimizing the batch size $\{\delta_{k}\}$, UAV position $\{\bs u_k\}$, and bandwidth allocation $\{B_{k}\}$.
%
\begin{subequations}
\begin{align}
\textbf{P1:}\ \ \  \min\limits_{\{\delta_{k}\},\{\bs u_k\},\{B_{k}\},\atop N, T_{\text{max}}} \ \ \ \ \ \ \ &N\cdot T_{\text{max}} \notag\\
\text{s.t.}\ \ \ \ \ \ \ \ \ \ \ \ &\Phi\left(\{\delta_{k}\},\{\bs u_k\}, N\right)\leq \epsilon, \label{equ:problemnewa}\\
&\mathbb E \left[T^{(n)}\left(\{\delta_k\}, \{\bs u_k\}, \{B_k\}\right)\right]\leq T_{\text{max}}, \forall n\in\mc N_{\epsilon}, \label{equ:problemnewb}\\
&q_{s,k}(\bs u_k)=q_{s,k'}(\bs u_{k'}), \forall k, k'\in\mc K, \label{equ:problemnewc}\\
&\theta_{s,k}(\bs u_k)\geq\theta_0, \forall k\in\mc K, \label{equ:problemnewd}\\
&\delta_k\in\mathbb Z^{+}, \forall k\in\mc K, \label{equ:problemnewe}\\
&\sum\limits_{k\in\mc K} B_k=B_c. \label{equ:problemnewf}
\end{align}
\end{subequations}
The objective function $N\cdot T_{\text{max}}$ is the total training time,
where $N$ is the number of training rounds required to guarantee the $\epsilon$-optimality gap, i.e., $\Exp{F\left(\mathbf w^{(N)}\right)-F(\mathbf w_*)}\leq\Phi\left(\{\delta_{k}\},\{\bs u_k\}, N\right)\leq\epsilon$ in \eqref{equ:problemnewa}, and $T_{\text{max}}$ is the per round latency.
Constraint in \eqref{equ:problemnewb} indicates that the average training time per round cannot exceed the delay requirement $T_{\text{max}}$, and $\mc N_{\epsilon}=\{1,\dots,N\}$ is the set of training rounds.
As suggested by Corollary \ref{corollary:Gsameq}, it is crucial to maintain a uniform successful sensing probability across the UAVs, so we enforce the constraints in \eqref{equ:problemnewc}.
The constraint on sensing quality \eqref{equ:problemnewd} has been discussed in Section \ref{sec:sense_analysis}.
We consider that the batch size $\{\delta_k\}$ are integer variables in \eqref{equ:problemnewe}.
Moreover, equation \eqref{equ:problemnewf} constrains the total bandwidth allocated to all UAVs as $B_c$ in Section \ref{sec:communication}.

\begin{figure}[t]
\begin{center}
\includegraphics[width=7.5cm, trim = 0cm 0.2cm 0cm 0cm, clip = true]{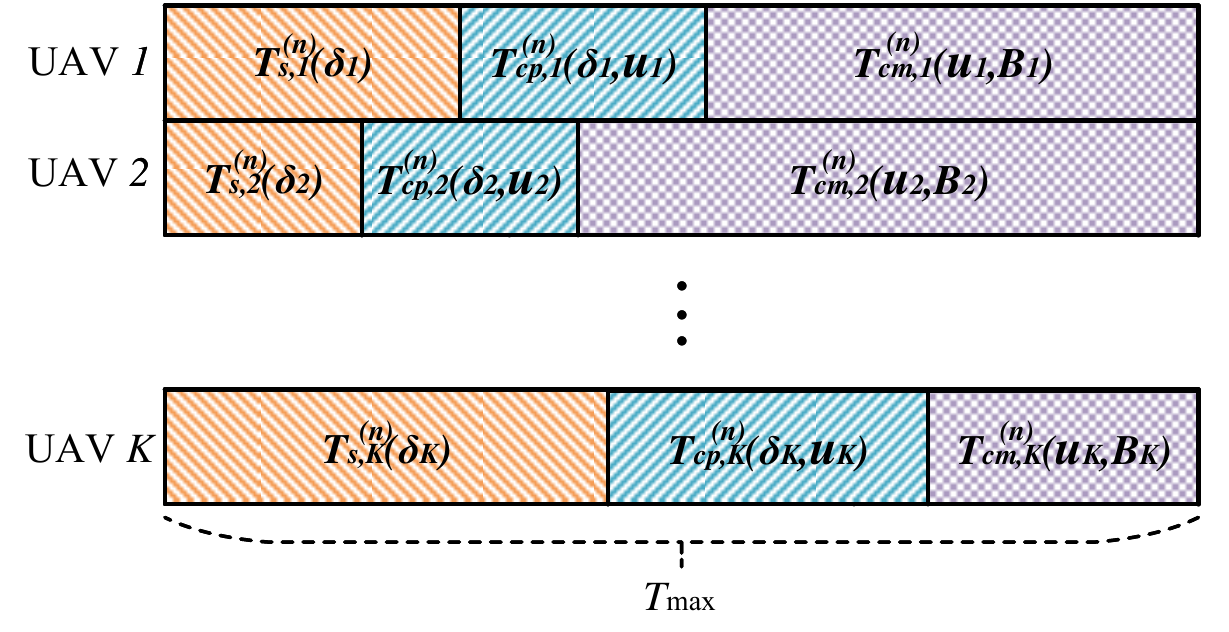}
\caption{\footnotesize{Latency illustration of UAVs that successfully sense the targets in round $n$. Sensing time $T_{s,k}^{(n)}(\delta_k)$, computation time $T_{cp,k}^{(n)}(\delta_k,\bs u_k)$, and communication time $T_{cm,k}^{(n)}(\bs u_k, B_k)$ of each UAV $k$ are indicated in orange, blue, and purple, respectively. The per round latency is $T_{\text{max}}= \max\limits_{k\in\mc K}\left\{T_{s,k}^{(n)}(\delta_k)+T_{cp,k}^{(n)}(\delta_k, \bs u_k)+T_{cm,k}^{(n)}(\bs u_k,B_k)\right\}$ based on \eqref{equ:round latency}.}}
\label{fig:tradeoff}
\end{center}
\end{figure}

The objective function and constraints \eqref{equ:problemnewa} and \eqref{equ:problemnewb} in \textbf{P1} are complicated by the coupling of $\{\delta_k\}$, $\{\bs u_k\}$, and $\{B_k\}$. 
Furthermore, the batch size $\delta_k$ of each UAV $k$ can only take integer values.
Therefore, \textbf{P1} is a mixed-integer non-convex problem and is thus very challenging to solve optimally.
To deal with this difficulty, we will utilize the alternating optimization technique in Section \ref{sec:train_time_minimize} to solve \textbf{P1} with any given $N$, in which the three decision variables are optimized alternately.
Moreover, we adopt a one-dimension search to find $N$ to achieve the minimum objective value, as summarized in Algorithm \ref{algo:BBPO} of Section \ref{sec:BBPO}.

Next, we show how the decision variables $\{\delta_k\}$, $\{\bs u_k\}$ and $\{B_k\}$ collectively affect the objective function $N\cdot T_{\text{max}}$ as follows.
\begin{remark}\label{remark:NT_tradeoff1}
\textbf{The batch size $\{\delta_k\}$ balances the number of training rounds $N$ and the per round latency $T_{\text{max}}$.}
\emph{ 
For UAVs that successfully sense targets, a larger $\{\delta_k\}$ indicates longer sensing and local computation time, since they are proportional to $\{\delta_k\}$ defined in \eqref{equ:sensing time} and \eqref{equ:local updating time}.
Thus, the $T_{\text{max}}$, including sensing, computation and communication time, will also increase.
However, large batch sizes can speed up the convergence of FEEL and reduce $N$. }
Therefore, the batch size $\{\delta_k\}$ must be properly designed to minimize the total training time.
\end{remark}
%
\begin{remark}\label{remark:NT_tradeoff2}
\textbf{The UAV position $\{\bs u_k\}$ also balances the number of training rounds $N$ and the per round latency $T_{\text{max}}$.}
\emph{
Specifically, when UAV $k$ is close to the target, the sensing elevation angle $\theta_{s,k}(\bs u_k)$ will increase.
Since the successful sensing probability $q_{s,k}(\bs u_k)$ increases with $\theta_{s,k}(\bs u_k)$, the UAV also have a higher probability of participating in the training. This will accelerate the convergence of FEEL and reduce $N$.
However, in this case, UAVs may be far away from the server, and the bandwidth allocated to each UAV will be reduced due to more UAVs participating in the training.
These factors will slow the uplink transmission rate and increase $T_{\text{max}}$.
Therefore, the UAV position $\{\bs u_k\}$ also needs to be properly designed to minimize the total training time.}
\end{remark}
%
\begin{remark}\label{remark:B_balance_T}
\textbf{The Bandwidth $\{B_k\}$ can minimize the per round latency $T_{\text{max}}$.}
\emph{Based on \eqref{equ:sensing time} and \eqref{equ:local updating time}, we have that sensing time $T_{s,k}^{(n)}(\delta_k)$ and local computation time $T_{cp,k}^{(n)}(\delta_k,\bs u_k)$ increase with $\delta_k$.
Since each UAV can choose different $\delta_k$ the number of generated samples, $T_{s,k}^{(n)}(\delta_k)$ and $T_{cp,k}^{(n)}(\delta_k,\bs u_k)$ will exhibit diversity among different UAVs.
We introduce bandwidth allocation in the system, let each UAV $k$ control the uplink transmission time $T_{cm,k}^{(n)}(\bs u_k, B_k)$ by adjusting $B_k$, to minimize $T_{\text{max}}$.
As shown in Fig. \ref{fig:tradeoff}, UAV $2$ spends less time on sensing and local computing, while UAV $K$ spends more time on them. That is, $T_{s,2}^{(n)}(\delta_2) + T_{cp,2}^{(n)}(\delta_2,\bs{u}_k) < T_{s,K}^{(n)}(\delta_K) + T_{cp,K}^{(n)}(\delta_2,\bs{u}_K)$.
In this way, we can allocate less bandwidth to UAV $2$ to achieve large $T_{cm,2}^{(n)}(\bs u_2, B_2)$ and allocate more bandwidth to UAV $K$ to achieve small $T_{cm,K}^{(n)}(\bs u_K, B_K)$, so that the total latency of UAV $2$ is equal to that of UAV $K$.
Thus, adjusting $\{B_k\}$ can minimize $T_{\text{max}}$.
}
\end{remark}
%

\section{Training Time Minimization} \label{sec:train_time_minimize}
In this section, we adopt the alternating optimization technique to solve \textbf{P1}.
We first divide \textbf{P1} into three subproblems, and then optimize $\{\delta_k\}$, $\{\bs u_k\}$, and $\{B_k\}$ in Section \ref{sec:bandwidth}, \ref{sec:senstime}, and \ref{sec:UAVposition}, respectively.
Overall, we propose the bandwidth, batch size, and position optimization (BBPO) scheme, which can compute a suboptimal solution of \textbf{P1} efficiently in Section \ref{sec:BBPO}.

\subsection{Bandwidth Allocation Optimization with Fixed Batch Size and Position} \label{sec:bandwidth}
We consider the subproblem of \textbf{P1} for optimizing the bandwidth allocation $\{B_k\}$ by assuming that batch size $\{\delta_k\}$, UAV position $\{\bs u_k\}$, and training round $N$ are fixed.
It follows from \eqref{equ:los_s} that the successful sensing probability $\left\{q_{s,k}(\bs u_k)\right\}$ are given. Given that $\{\delta_k\}$ and $\{\bs u_k\}$ satisfy the constraints \eqref{equ:problemnewa}, and \eqref{equ:problemnewc}-\eqref{equ:problemnewe}, \textbf{P1} reduces to
\begin{subequations}
\begin{align}
\!\!\textbf{P2:}\ \ \ \min\limits_{\{B_{k}\}} \ \ \ \ \ & T_{\text{max}} \notag\\
\text{s.t.}\ \ \ \ \ \
&\Exp{\!T_0\!\cdot \!\delta_k\!+\!\mathbf 1_{k}^{(n)}(\bs u_k) \!\cdot\!\frac {\delta_k \xi}{f_{cpu}}\!+\!\mathbf 1_{k}^{(n)}(\bs u_k) \!\cdot\!\frac {D_0}{r_k(\bs u_k,B_k)}\!}\!\leq\! T_{\text{max}}, \forall k\in\mc K, \forall n\in\mc N_{\epsilon},\label{equ:problemnew1a}\\
&\text{contraint \eqref{equ:problemnewf}},\nonumber
\end{align}
\end{subequations}
where \eqref{equ:problemnew1a} is the same as \eqref{equ:problemnewb}.
Specifically, due to \eqref{equ:q_indicator}, constraints \eqref{equ:problemnew1a} can be simplified to
\begin{align}
T_0\cdot \delta_k+q_s \cdot\frac {\delta_k \xi}{f_{\text{cpu}}}+q_s \cdot\frac {D_0}{r_k(\bs u_k,B_k)}\leq T_{\text{max}}, \forall k\in\mc K, \label{equ:T_simplify}
\end{align}
where $q_{s,k}(\bs u_k)=q_s, \forall k\in \mc K$.
Since the left hand side of constraints \eqref{equ:T_simplify} are convex with respect to $\{B_k\}$, \textbf{P2} is a standard convex problem. 
Therefore, it can be efficiently solved by standard convex optimization tools such as CVXPY \cite{CVXPY}.


\subsection{Batch Size Optimization with Fixed Position and Bandwidth Allocation} \label{sec:senstime}
We consider the subproblem of \textbf{P1} for optimizing batch size $\{\delta_k\}$ by assuming that $\{\bs u_k\}$, $\{B_k\}$, and training round $N$ are fixed.
Since the successful sensing probability $q_{s,k}(\bs u_k)=q_s, \forall k\in\mc K$ is determined by $\{\bs u_k\}$, parameter $q_s$ is also given.
In this case, given that $\{u_k\}$ and $\{B_k\}$ satisfy the constraints \eqref{equ:problemnewc}, \eqref{equ:problemnewd}, and \eqref{equ:problemnewf}, \textbf{P1} reduces to
\begin{subequations}
\begin{align}
\textbf{P3:}\ \ \ \min\limits_{\{\delta_{k}\}} \ \ \ \ \ \ \ & T_{\text{max}} \notag\\
\text{s.t.}\ \ \ \ \ \ \ \ &\delta_k>0, \forall k\in\mc K, \label{equ:problemnew2a}\\
&\text{constraints \eqref{equ:problemnewa}, \eqref{equ:T_simplify}}, \nonumber
\end{align}
\end{subequations}
where we relax the value of $\{\delta_k\}$ from integer in \eqref{equ:problemnewe} to real number $\delta_k>0, \forall k\in\mc K$,
and \eqref{equ:problemnewb} is simplified as \eqref{equ:T_simplify}.
Next, we characterize the optimal solution of \textbf{P3} as follows.

\begin{lemma} \label{lemma:constrain_T_delta}
At the optimal $\{\delta_k^*\}$ of \textbf {P3}, equality constraint in \eqref{equ:T_simplify} holds. That is,
\begin{align}
T_0\cdot \delta_k^*+q_s \cdot\frac {\delta_k^* \xi}{f_{\text{cpu}}}+q_s \cdot\frac {D_0}{r_k(\bs u_k,B_k)} = T_{\text{max}}, \forall k\in\mc K.  \label{equ:constrain_T_delta}
\end{align}
\end{lemma}
Therefore, each $\delta_k^*$ can be represented as a function of $T_{\text{max}}$, i.e.,
\begin{align}
\delta_k^*(T_{\text{max}})=\frac{T_{\text{max}}-\frac {q_sD_0}{r_k(\bs u_k,B_k)}}{T_0+\frac {q_s \xi}{f_{\text{cpu}}}}. \label{equ:delta2}
\end{align}
%
According to \eqref{equ:Fbound}, as $\Phi\left(\left\{\delta_{k}(T_{\text{max}})\right\},\{\bs u_k\}, N\right)$ increases, $T_{\text{max}}$ decreases.
To minimize $T_{\text{max}}$, we consider $\Phi\left(\left\{\delta_{k}(T_{\text{max}})\right\},\{\bs u_k\}, N\right)=\epsilon$ in \eqref{equ:problemnewa}.
Based on \eqref{equ:Fbound}, \eqref{equ:Gsameq}, and \eqref{equ:problemnewa}, we have
\begin{align}
\displaystyle\frac {\eta} {K}\sum\limits_{k\in\mc K}\frac{\sigma_k^2}{\delta_k^*(T_{\text{max}})}
+\frac{2L\eta^2}{K^2 \chi q_s}\sum\limits_{k\in\mc K}\displaystyle\frac{\sigma_k^2}{\delta_k^*(T_{\text{max}})}
+J
= \frac{(1-A)(\epsilon-\lambda A^N)}{1-A^N}, \label{equ:T2_extension}
\end{align}
where $J=\left(\frac {\eta} {K}+\frac{4L\eta^2}{K^2 \chi q_s}+\frac{4L\eta^2}{K q_s^{K-2}}\right)\sum\limits_{k\in\mc K}\Lambda_{k}^2$, $A=1-\mu\eta(1-4L\eta)$, and $\lambda=\Exp{F\left(\mathbf w^0\right)-F(\mathbf w_*)}$.
It can be transformed into a polynomial equation with respect to $T_{\text{max}}$ and solved efficiently.

\subsection{Position Optimization with Fixed Batch Size and Bandwidth Allocation} \label{sec:UAVposition}
We consider the subproblem of \textbf{P1} for optimizing UAV position $\{\bs u_k\}$ by assuming that $\{\delta_k\}$, $\{B_k\}$, and training round $N$ are fixed.
In this case, \textbf{P1} given that $\{\delta_k\}$ and $\{B_k\}$ satisfy the constraints \eqref{equ:problemnewe} and \eqref{equ:problemnewf} reduces to
\begin{subequations}
\begin{align}
\textbf{P4:}\ \ \ \min\limits_{\{\bs u_k\}} \ \ \ \ \ \ \ &T_{\text{max}} \notag\\
\text{s.t.}\ \ \ \ \ \ \ \ &\text{constraints \eqref{equ:problemnewa}, \eqref{equ:T_simplify}, \eqref{equ:problemnewc}, \eqref{equ:problemnewd}}. \nonumber
\end{align}
\end{subequations}
Constraints \eqref{equ:problemnewb} in \textbf{P1} can be simplified as \eqref{equ:T_simplify}.
Below, we aim to transform the constraints on $\{\bs u_k\}$ of \textbf{P4} into constraints on $q_s$, as this will simplify our problem optimization. Since \eqref{equ:problemnewc} indicates $q_{s,k}(\bs u_k)=q_s, \forall k\in\mc K$, we omit the transformation of \eqref{equ:problemnewc}.

\subsubsection{Constraint \eqref{equ:problemnewa}}  \label{sec:conv1}
Based on \eqref{equ:Fbound}, \eqref{equ:Gsameq}, and \eqref{equ:problemnewa}, we have
\begin{align}
\displaystyle\frac {\eta} {K}\sum\limits_{k\in\mc K}\left(\displaystyle\frac{\sigma_k^2}{\delta_k}+\Lambda_k^2\right)
+\frac{2L\eta^2}{K^2 \chi q_s}\sum\limits_{k\in\mc K}\left(\displaystyle\frac{\sigma_k^2}{\delta_k}+2\Lambda_k^2\right)
+\frac{4L\eta^2}{K q_{s}^{K-2}}\sum\limits_{k\in\mc K}\Lambda_{k}^2
\leq \frac{(1-A)(\epsilon-\lambda A^N)}{1-A^N}. \label{equ:phi_q_extension}
\end{align}
Since $q_s\geq q_{s,\text{min}}$, \eqref{equ:phi_q_extension} can be rewritten as
\begin{align}
q_s\geq
\frac{2L\eta^2}{K^2 \chi}\sum\limits_{k\in\mc K}\left(\displaystyle\frac{\sigma_k^2}{\delta_k}\!+\!2\Lambda_k^2\right)
\!\!\Bigg /\!\!
\left[\frac{(1\!-\!A)(\epsilon\!-\!\lambda A^N)}{1\!-\!A^N}
\!-\!\displaystyle\frac {\eta} {K}\!\sum\limits_{k\in\mc K}\!\left(\displaystyle\frac{\sigma_k^2}{\delta_k}\!+\!\Lambda_k^2\right)
\!-\!\frac{4L\eta^2}{K q_{s,\text{min}}^{K-2}}\sum\limits_{k\in\mc K}\Lambda_{k}^2\right]. \label{equ:phi_q}
\end{align}
%

\subsubsection{Constraint \eqref{equ:problemnewd}} \label{sec:conv2}
According to \eqref{equ:los_s}, the successful sensing probability increases as the sensing elevation angle $\theta_{s,k}(\bs u_k)$ increases.
Based on \eqref{equ:problemnewd}, we have
\begin{align}
q_{s}\geq\frac 1 {1+\psi \exp(-\zeta[\theta_{0}-\psi])}. \label{equ:uv_q}
\end{align}
\begin{figure}[t]
\begin{center}
\includegraphics[width=7.5cm, trim = 0cm 0cm 0cm 0cm, clip = true]{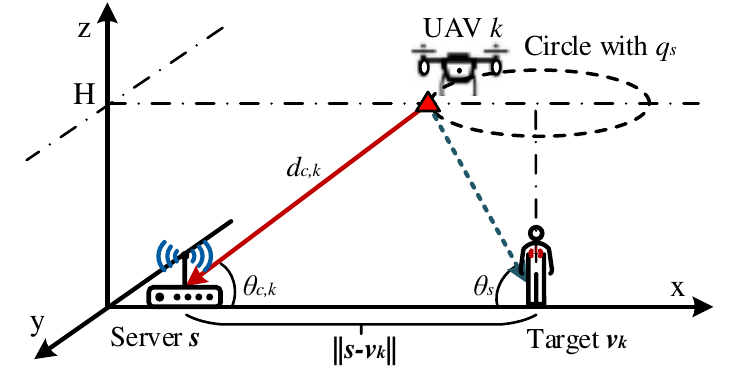}
\caption{\footnotesize{An illustration of the optimal UAV position with given $q_s$. The optimal position is marked with a red triangle.}}
\label{fig:q_u}
\end{center}
\end{figure}
%
\subsubsection{Relationship between $q_s$ and $r_k(\bs u_k,B_k)$ in \eqref{equ:T_simplify}} \label{sec:conv3}
After the above constraint transformations, we can see that only $r_k(\bs u_k,B_k)$ in \textbf{P4} contains $\bs u_k$.
According to the expression of $r_k(\bs u_k,B_k)$ in \eqref{equ:UAV-server}, we have the following proposition.
\begin{proposition} \label{proposition:q_u}
For any given $q_s$ and $B_k$, the optimal position $\bs{u}_k^* = (x_{u,k}^*, y_{u,k}^*, H)$ for each UAV $k$ to maximize $r_k(\bs u_k,B_k)$ is derived from
$x_{u,k}^* = x_{v,k}-\frac {H(x_{v,k}-x_s)}{\tan \theta_s\Vert\bs s-\bs v_k\Vert}$ and $y_{u,k}^* = y_{v,k}-\frac {H(y_{v,k}-x_s)}{\tan \theta_s\Vert\bs s-\bs v_k\Vert}$,
where $\theta_s$ is the corresponding sensing elevation angle of $q_s$ in \eqref{equ:los_s}.
\end{proposition}
\begin{proof} \label{proof:q_u}
From solid geometry, position $\bs{u}_k^*$ leads to the minimum UAV-server distance $d_{c,k}$, and the maximum communication elevation angle $\theta_{c,k}$. Due to \eqref{equ:UAV-server}, we have the maximum $r_k(\bs u_k,B_k)$.
\end{proof}

We mark the optimal position with a red triangle in Fig. \ref{fig:q_u}.

\begin{lemma} \label{lemma:r_s}
According to $\bs{u}_k^*$ with given $q_s$ in Proposition \ref{proposition:q_u}, transmission rate $r_k(\bs u_k,B_k)$ can be written as $\hat r_k(q_s, B_k)$, which is a monotonically decreasing function of $q_s$.
\end{lemma}
\begin{proof}
From Fig. \ref{fig:q_u}, we have the relationship between $\theta_{c,k}$ and $\theta_s$ as
%
$\theta_{c,k} \!= \!\tan^{-1}  \!\!\left(\!\!\frac{H}{\Vert\bs s-\bs v_k\Vert-\frac H{\tan\theta_s}} \!\!\right)\times\frac{180^{\circ}}{\pi},$
%
which is derived by $\frac {H}{\tan\theta_{c,k}}+\frac{H}{\tan\theta_s}=\Vert\bs s-\bs v_k\Vert$.
Moreover, the relationship between $d_{c,k}^2$ and $\theta_s$ is
%
$d_{c,k}^2 = \left(\Vert\bs s-\bs v_k\Vert-\frac H{\tan\theta_s}\right)^2+H^2.$
%
Based on \eqref{equ:los_s}, the sensing elevation angle $\theta_s$ can be represented as a function of $q_s$, i.e.,
%
$\theta_{s} = -\frac 1 {\zeta}\left[\ln\left(\frac 1 {q_s}-1\right)-\ln \psi\right]+\psi.$
%
Substituting the above equations into \eqref{equ:UAV-server}, we have $\hat r_k(q_s, B_k)$. Also, it decreases monotonically as $q_s$ increases.
\end{proof}

Based on Proposition \ref{proposition:q_u} and Lemma \ref{lemma:r_s}, \textbf{P4} can be reformulated as
\begin{subequations}
\begin{align}
\textbf{P4':}\ \ \ \min\limits_{q_s} \ \ \ \ \ \ \ &T_{\text{max}} \notag\\
\text{s.t.}\ \ \ \ \ \ \ &T_0\cdot \delta_k+q_s \cdot\frac {\delta_k \xi}{f_{\text{cpu}}}+q_s \cdot\frac {D_0}{\hat r_k(q_s, B_k)}\leq T_{\text{max}}, \forall k\in\mc K, \label{equ:problemnew31a}\\
&0< q_s \leq 1, \label{equ:problemnew31b}\\
&\text{constraints \eqref{equ:phi_q}, \eqref{equ:uv_q}.} \nonumber
\end{align}
\end{subequations}
Since $\hat r_k(q_s, B_k)$ is a monotonically decreasing function of $q_s$, we can achieve the minimum $T_{\text{max}}$ by minimizing $q_s$ from \eqref{equ:problemnew31a}.
We first find the minimum $q_s^*$ that satisfies constraints \eqref{equ:phi_q}, \eqref{equ:uv_q}, and \eqref{equ:problemnew31b}.
Then, based on \eqref{equ:problemnew31a} we have $T_{\text{max}}=\max\limits_{k\in\mc K}\left\{T_0\cdot \delta_k+q_s^* \cdot\frac {\delta_k \xi}{f_{\text{cpu}}}+q_s^* \cdot\frac {D_0}{\hat r_k(q_s^*, B_k)}\right\}$.
Finally, the optimal UAV position $\{\bs u_k^*\}$ can be derived from Proposition \ref{proposition:q_u}.

\begin{algorithm}[t]
  \caption{Iterative Bandwidth, Batch Size, and Position Optimization (BBPO)}
  \label{algo:BBPO}
  \KwIn{UAV set $\mc K$; target location set $\{\bs v_k\}$; UAV transmission power $\{p_{c,k}\}$; the minimum sensing elevation angle $\theta_0$; accuracy threshold $\tau$; the minimum training round $N_{\text{min}}\!=\!\log_{A}\left(\frac{\epsilon(1-A)-G_{\text{max}}}{\lambda(1-A)-G_{\text{max}}}\right)$; the maximum training round $N_{\text{max}}$;}
  \For{$n=N_{\text{min}}:N_{\text{max}}$}{
  Set the iteration index of the following three steps as $i=0$\;
  Initialize UAV batch size $\{\delta_k^{(i)}\}$ to satisfy the constraints in \eqref{equ:problemnewe} and UAV position $\{\bs u_k^{(i)}\}$ to satisfy the constraints in \eqref{equ:problemnewc} and \eqref{equ:problemnewd}\;
  \Repeat
  {\rm {$|T_{\text{max}}^{(i)}(n)-T_{\text{max}}^{(i-1)}(n)|\leq \tau$}}
  {
  Fix $\{\delta_k^{(i)}\}$ and $\{\bs u_k^{(i)}\}$, optimize bandwidth $\{B_k^{(i)}\}$ by solving \textbf{P2} in Section \ref{sec:bandwidth}\;
  Fix $\{\bs u_k^{(i)}\}$ and $\{B_k^{(i)}\}$, optimize batch size $\{\delta_k^{(i)}\}$ by solving \textbf{P3} in Section \ref{sec:senstime}\;
  Fix $\{\delta_k^{(i)}\}$ and $\{B_k^{(i)}\}$, optimize position $\{\bs u_k^{(i)}\}$ by solving \textbf{P4} in Section \ref{sec:UAVposition}\;
  Let $T_{\text{max}}^{(i)}(n)$ be the objective value of \textbf{P4}\;
  Update $i=i+1$\;
  }
  }
  Find the optimal training round $N^*=\argmin\limits_{n\in\{N_{\text{min}},\dots,N_{\text{max}}\}}n\cdot T_{\text{max}}^{(i)}(n)$\;
    \KwOut{$\{\delta^*_k\}=\{\lceil\delta_k^{(i)}\rceil\}$, $\{\bs u^*_k\}=\{\bs u_k^{(i)}\}$, and $\{B^*_k\}=\{B_k^{(i)}\}$.}

\end{algorithm}
%

\subsection{Iterative Bandwidth, Batch Size, and Position Optimization (BBPO)}\label{sec:BBPO}
Based on the solutions of the three subproblems in the above sections, we propose the iterative algorithm for \textbf{P1}, summarized in Algorithm \ref{algo:BBPO}\footnote{The parameters $A$ and $\lambda$ are defined in \eqref{equ:T2_extension}.
The minimum training round $N_{\text{min}}$ is derived by \eqref{equ:Fbound} with constraint $\lambda(1-A)-G_{\text{max}}>0$, where $G_{\text{max}}$ is derived by \eqref{equ:Gsameq} with $q_s=1$ and $\delta_k=\delta_{\text{max}}, \forall k\in\mc K$. The maximum training round $N_{\text{max}}$ depends on the maximum tolerance training time and is thus a given parameter.}.
We first employ a one-dimension search for $N$. With each given $N$, we solve sub-problems \textbf{P2}, \textbf{P3}, and \textbf{P4} alternately. 
Since $N$ is searched linearly and \textbf{P2}-\textbf{P4} can be efficiently solved, Algorithm \ref{algo:BBPO} can efficiently solve \textbf{P1}.

\section{Simulation Results} \label{sec:simulation}
In this section, we present the numerical results to evaluate the performance of our proposed BBPO scheme.
We first describe the basic simulation settings in Section \ref{sec:parameters}. Then, we discuss the baseline schemes in Section \ref{sec:benchmark scheme}. We demonstrate and compare the BBPO scheme's performance with the baseline schemes in Section \ref{sec:performance evaluation}.


%
\begin{table}[t]
\caption{System Parameters}
\label{tab:system parameters}
\centering
\begin{tabular}{llll}
\hline
{\textbf{Parameter}} & {\textbf{Value}}&{\textbf{Parameter}} & {\textbf{Value}}\\
\hline
\hline
{Communication bandwidth, $B_c$} & $6$ MHz  \cite{Bandwidth} &{Sensing bandwidth, $B_s$} & $10$ MHz \cite{ISCC_liu}\\
{Communication transmit power, $p_c$} & $20$ dBm \cite{Bandwidth} &{Sensing transmit power, $p_s$} & $30$ dBm \cite{ISCC_liu}\\
{Carrier frequency, $f_c$} & $60$ GHz \cite{ISCC_liu} &{Chirp duration, $T_p$} & $10$ \textmu s \cite{ISCC_liu}\\
{Noise power spectral density, $N_0$} & $-174$ dBm/Hz \cite{ISCC_liu} & {Chirp numbers per frame, $M$} & $25$ \cite{ISCC_liu}\\
{Path-loss exponent, $\alpha$} & $2$  \cite{3GPP2017} & {Unit sensing time, $T_0$} & $0.5$ s \cite{ISCC_liu}\\
{Path-loss for LOS, NLOS, $\eta_1, \eta_2$} & $3, 23$ dB \cite{3GPP2017} & {CPU frequency, $f_{\text{cpu}}$} & $5\times 10^8$ cycles/s \cite{ISCC_liu}\\
{Environment parameters, $\psi, \zeta$} & $11.95, 0.14$ \cite{3GPP2017} & {CPU cycles for one sample, $\xi$} & $2.5\times 10^7$ \cite{ISCC_liu}\\
\hline
\end{tabular}
\end{table}
%

\subsection{Simulation Settings} \label{sec:parameters}

\subsubsection{UAV-assisted FEEL system}
We consider a UAV-assisted federated learning system with an edge server and $K=8$ ISAC UAVs.
The UAVs are randomly distributed within a circular area of radius $300$ meters and are geometrically separated at different positions for sensing.
Moreover, we define the minimum sensing elevation angle as $\theta_0=70^{\circ}$, corresponding to PSNR $=30$ dB in Fig. \ref{fig:PSNR}. Other simulation parameters are listed in Table~\ref{tab:system parameters}.

\subsubsection{Human motion recognition}
We apply the wireless sensing simulator in \cite{sense_platform} to simulate various human motions and generate human motion datasets. In the simulation, we aim to identify five different human motions, i.e., child walking, child pacing, adult walking, adult pacing, and standing \cite{ISCC_liu}.
Examples of data samples of each human motion are shown in Fig. \ref{fig:quality}.

\subsubsection{Learning model}
We apply widely used ResNet-10 ($4,900,677$ model parameters) with batch normalization as the classifier model. The learning rate is set to $\eta=0.03$.

%
\begin{figure}[t]
\begin{center}
\setlength{\abovecaptionskip}{-0.2cm}
\includegraphics[width=9cm, trim = 0cm 0cm 0cm 1.5cm, clip = true]{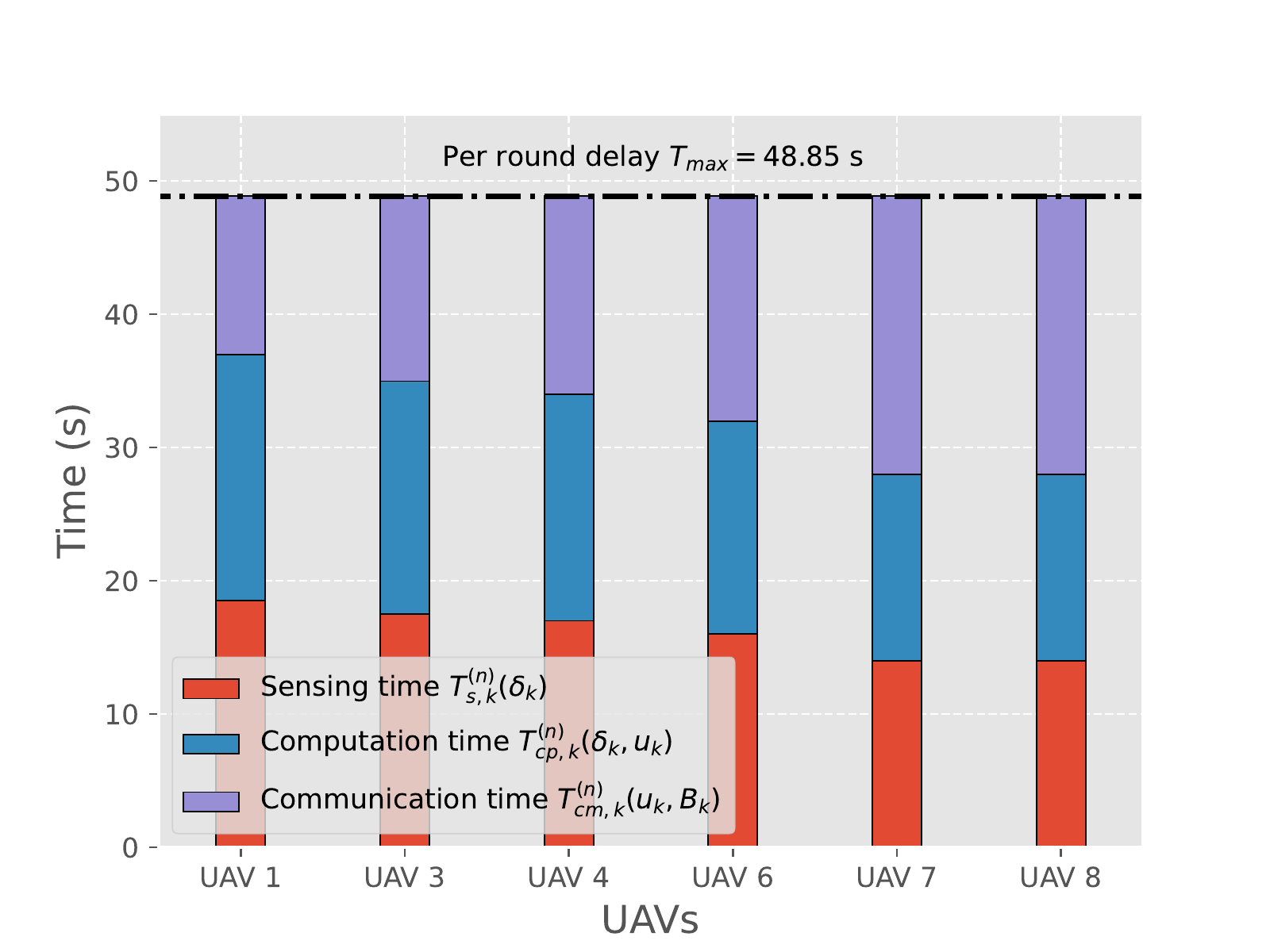}
\caption{\footnotesize{An illustration of sensing time $T_{s,k}^{(n)}(\delta_k)$, computation time $T_{cp,k}^{(n)}(\delta_k,\bs u_k)$, and communication time $T_{cm,k}^{(n)}(\bs u_k, B_k)$ of participating UAVs in round $n=700$ under our proposed BBPO scheme. UAV $2$ and UAV $5$ do not participate in this training round, since they do not sense the targets successfully.}}
\label{fig:Tmax}
\end{center}
\end{figure}

\subsection{Baseline Schemes} \label{sec:benchmark scheme}
%
\begin{itemize}
    \item \textbf{Baseline 1:} (Det-UAVposition) This scheme considers that the UAV position $\{\bs u_k\}$ is given. The batch size $\{\delta_k\}$ and bandwidth $\{B_k\}$ are optimized as discussed in Section \ref{sec:bandwidth} and Section \ref{sec:senstime} in our proposed scheme. By comparing with this baseline, we evaluate the validity of the UAV position design in our proposed scheme.
  \item \textbf{Baseline 2:} (Eq-Bandwidth) This scheme uses the uniform bandwidth $B_k = B_c/K$ for uplink transmission in Section \ref{sec:bandwidth}. The batch size $\{\delta_k\}$ and UAV position $\{\bs u_k\}$ are optimized as discussed in Section \ref{sec:senstime} and Section \ref{sec:UAVposition} in our proposed scheme. By comparing with this baseline, we evaluate the validity of bandwidth allocation in our proposed scheme.
    \item \textbf{Baseline 3:} (Eq-Batchsize) This scheme considers the uniform batch size $\delta_k = \delta, \forall k\in\mc K$ in Section \ref{sec:senstime}. The bandwidth $\{B_k\}$ and UAV position $\{\bs u_k\}$ are optimized as discussed in Section \ref{sec:bandwidth} and Section \ref{sec:UAVposition} in our proposed scheme. By comparing with this baseline, we evaluate the validity of batch size design in our proposed scheme.
  \item \textbf{Baseline 4:} (BBPO-Ideal) The only difference between Baseline 4 and the proposed BBPO scheme is that this scheme assumes that all UAVs always successfully sense the targets and that all UAVs participated in each training round. Specifically, we have the successful sensing probability $q_s=1$ in \eqref{equ:problemnewa}, \eqref{equ:problemnewb}, and \eqref{equ:problemnewe}. The optimization of batch size $\{\delta_k\}$, UAV position $\{\bs u_k\}$, and bandwidth $\{B_k\}$ is the same as our proposed scheme. This baseline is an ideal situation for sensing.
\end{itemize}

\begin{figure*}[t]
\centering
\setlength{\abovecaptionskip}{-0.2cm}
\begin{minipage}[t]{0.45\linewidth}
\includegraphics[width=8cm, trim = 0cm 0cm 0cm 1.2cm, clip = true]{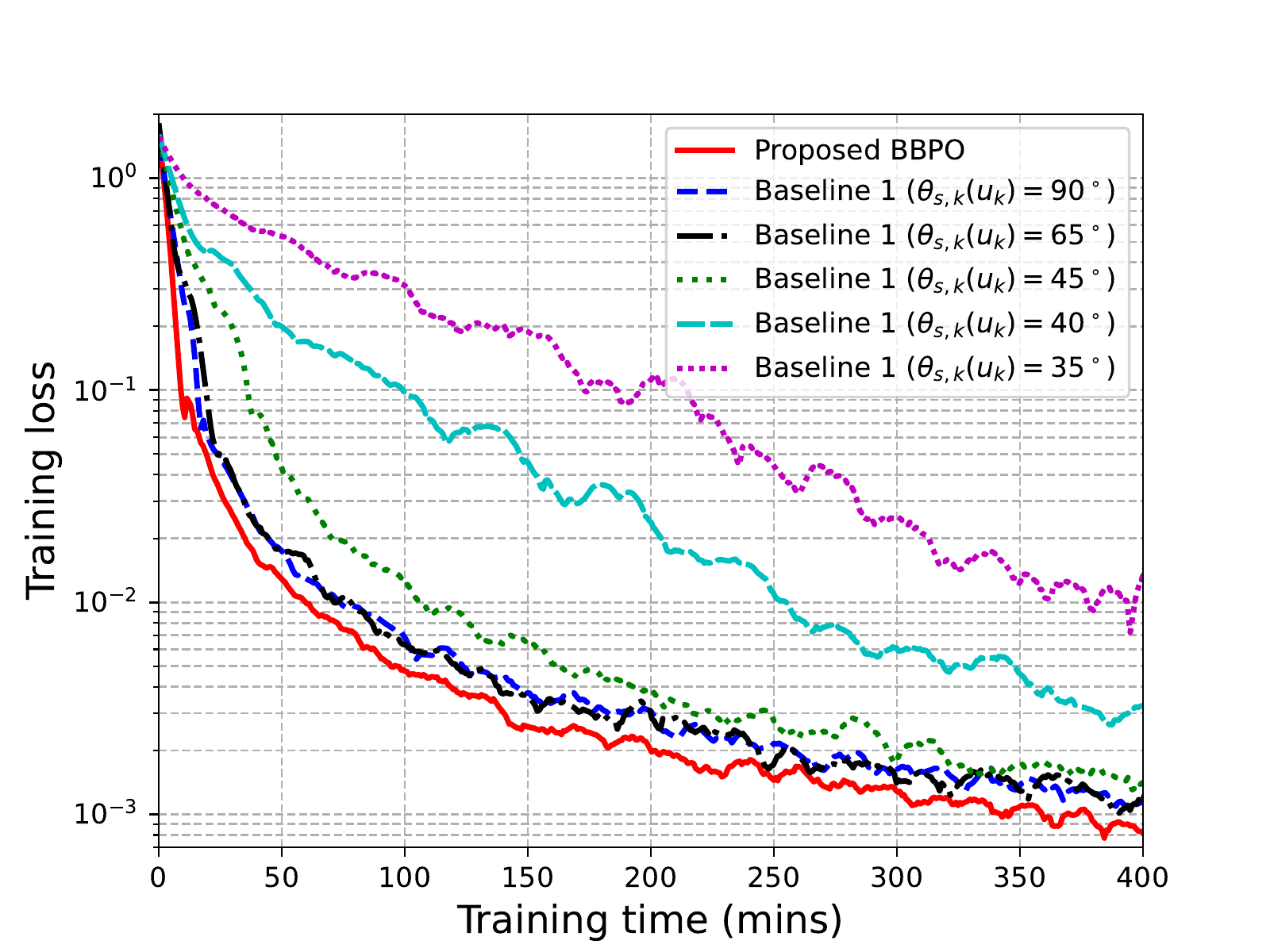}
\caption{\footnotesize{Training loss versus training time under the BBPO scheme and Baseline 1.}}
\label{fig:loss_BBPO}
\end{minipage}
\quad
\begin{minipage}[t]{0.45\linewidth}
\includegraphics[width=8cm, trim = 0cm 0cm 0cm 1.2cm, clip = true]{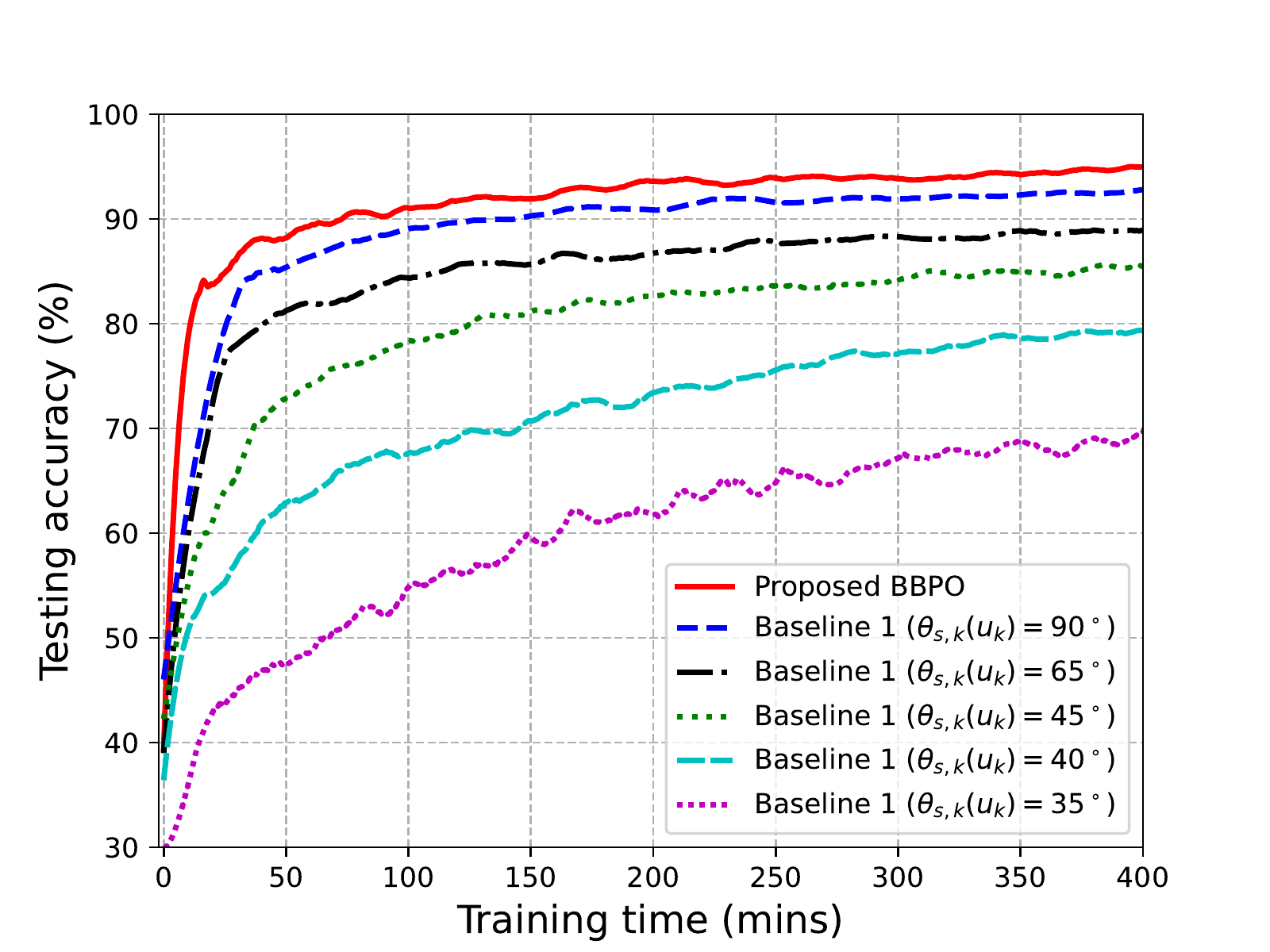}
\caption{\footnotesize{Testing accuracy versus training time under the BBPO scheme and Baseline 1.}}
\label{fig:accuracy_BBPO}
\end{minipage}
\end{figure*}
%

\subsection{Performance Evaluation} \label{sec:performance evaluation}

\subsubsection{\textbf{Balance effect of decision variables on $T_{\text{max}}$}} \label{sec:balance Tmax}
In Fig. \ref{fig:Tmax}, we plot the per round delay, which is the sum of the sensing, computation, and communication time, of each participating UAV under the BBPO scheme.
We can observe that the UAVs with longer sensing time $T_{s,k}^{(n)}(\delta_k)$ will also have longer computation time $T_{cp,k}^{(n)}(\delta_k,\bs u_k)$.
This is because $T_{s,k}^{(n)}(\delta_k)$ and $T_{cp,k}^{(n)}(\delta_k,\bs u_k)$ are proportional to batch size $\delta_k$.
Moreover, each UAV $k$ can control the uplink transmission time $T_{cm,k}^{(n)}(\bs u_k, B_k)$ by adjusting its $B_k$, to minimize  $T_{\text{max}}$.
This observation confirms out discussion in Remark \ref{remark:B_balance_T} (illustrated in Fig. \ref{fig:Tmax}) that all participating UAVs achieve the same $T_{\text{max}}=46.78$ s.

\subsubsection{\textbf{Impact of UAV position $\{\bs u_k\}$ on training performance}} \label{sec:UAV position performance}
In Fig. \ref{fig:loss_BBPO} and Fig. \ref{fig:accuracy_BBPO}, we evaluate the training performance of the proposed BBPO scheme by comparing with Baseline 1 (Det-UAVposition) under different UAV positions.
We use sensing elevation angle $\theta_{s,k}(\bs u_k)$ to reflect different UAV positions in Baseline 1. For instance, Baseline 1 ($\theta_{s,k}(\bs u_k)=35^\circ$) represents that current UAV position $\{\bs u_k\}$ corresponds to sensing elevation angle $\theta_{s,k}(\bs u_k)=35^\circ, \forall k\in\mc K$.

In Fig. \ref{fig:loss_BBPO}, we plot the training loss versus the total training time.
We can observe that the training loss decreases with training time under the BBPO scheme and Baseline 1, consistent with the analysis in Section \ref{sec:theorm_results}.
It also shows that the BBPO scheme achieves the fastest convergence rate compared to other schemes, and the convergence rate of Baseline 1 ($\theta_{s,k}(\bs u_k)=35^\circ$) is the slowest.
This is because the smaller $\theta_{s,k}(\bs u_k)$ indicates the further distance between the UAV and the sensing target, resulting in poor quality (i.e., lower PSNR) of data samples.
It also indicates the smaller successful sensing probability $q_s(\bs u_k)$, resulting in fewer UAVs participating in training. According to \eqref{equ:Fbound}, we need more training rounds $N$ to achieve the specified optimality gap.
However, this does not mean that a larger $\theta_{s,k}(\bs u_k)$ will lead to a better performance. As UAVs move closer to the target but farther away from the server, the uplink transmission time and per round latency $T_{\text{max}}$ will increase.
Since we consider the trade-off between $N$ and $T_{\text{max}}$ in our BBPO scheme. It achieves a faster convergence rate than other baseline schemes.
In Fig. \ref{fig:accuracy_BBPO}, we plot the testing accuracy under two schemes against the total training time. We observe that the BBPO scheme achieves the highest testing accuracy compared to other baseline schemes. The reason is the same as above.

\begin{figure*}[t]
\hspace{-0.2cm}
\centering
\setlength{\abovecaptionskip}{-0.2cm}
\begin{minipage}[t]{0.45\linewidth}
\includegraphics[width=8cm, trim = 0cm 0cm 0cm 1cm, clip = true]{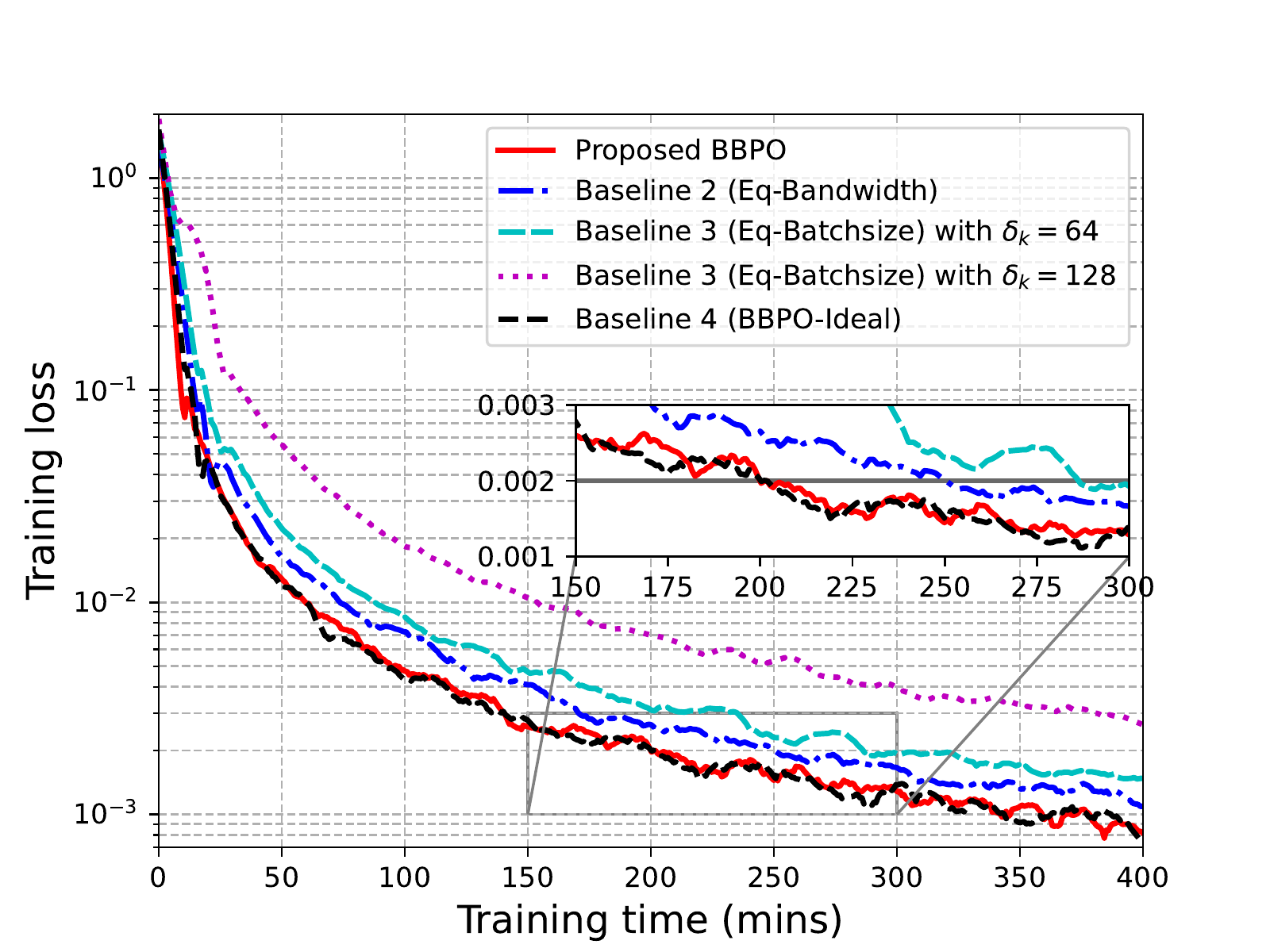}
\caption{\footnotesize{Training loss versus training time under different baseline schemes.}}
\label{fig:loss_benchmark}
\end{minipage}
\quad
\begin{minipage}[t]{0.45\linewidth}
\includegraphics[width=8cm, trim = 0cm 0cm 0cm 1cm, clip = true]{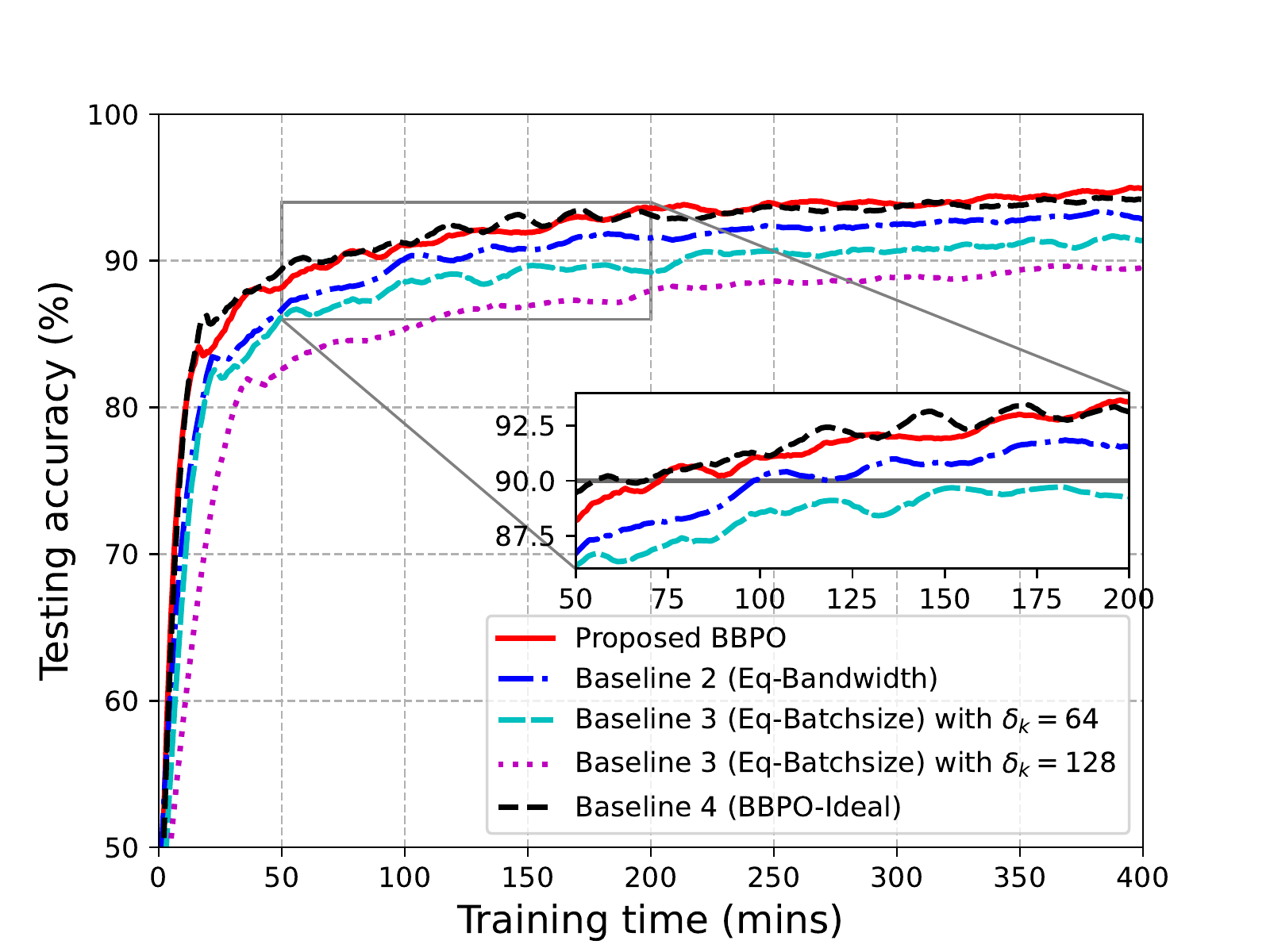}
\caption{\footnotesize{Testing accuracy versus training time under different baseline schemes.}}
\label{fig:accuracy_benchmark}
\end{minipage}
\end{figure*}
%

\subsubsection{\textbf{Impact of bandwidth $\{B_k\}$ and sensing time $\{\delta_k\}$ on training performance}} \label{sec:UAV position performance}

Fig. \ref{fig:loss_benchmark} and Fig. \ref{fig:accuracy_benchmark} show the training performance under four schemes. Compared with Baseline 2 (Eq-Bandwidth) and Baseline 3 (Eq-Batchsize), the proposed BBPO scheme achieves the best convergence rate and testing accuracy performance.
Due to the per round latency constraint \eqref{equ:problemnewb}, we try to adaptively assign different bandwidths to each UAV to minimize $T_{\text{max}}$ under the BBPO scheme.
However, in Baseline 2, all UAVs use a uniform bandwidth, which increases $T_{\text{max}}$ and worsens the training performance.
In Baseline 3, we consider that the UAVs have the uniform batch size $\delta_k=64$ or $\delta_k=128, \forall k\in\mc K$. Since a larger $\delta_k$ reduces the number of training rounds $N$ but increases $T_{\text{max}}$.
Therefore, the BBPO scheme that jointly optimizes $\{\delta_k\}$ and $\{B_k\}$ outperforms Baseline 2 and Baseline 3 in both Fig. \ref{fig:loss_benchmark} and Fig. \ref{fig:accuracy_benchmark}.
Specifically, in Fig. \ref{fig:loss_benchmark} under the same training loss $\epsilon=0.002$, the training time of the BBPO scheme is $20\%$ less than Baseline 2, and $33.33\%$ less than Baseline 3 with $\delta_k=64, \forall k\in\mc K$.
In Fig. \ref{fig:accuracy_benchmark}, to achieve a testing accuracy of $90\%$, the training time of the BBPO scheme is $30\%$ less than Baseline 2 and $60\%$ less than Baseline 3 with $\delta_k=64, \forall k\in\mc K$.
Moreover, it achieves almost the same performance as Baseline 4.

\section{Conclusion} \label{sec:conclusion}
In this paper, we investigated the UAV deployment design and resource allocation in UAV-assisted federated edge learning (FEEL) for improving training performance. 
Specifically, we studied the impact of UAV deployment on sensing quality for human motion recognition. We identified a threshold for the sensing elevation angle that produces satisfactory quality of data samples.
Then, we derived an upper bound on the UAV-assisted FEEL training loss as a function of the successful sensing probability, and showed that the negative impact of data heterogeneity can be reduced if UAVs have a uniform successful sensing probability.
Based on the convergence analysis, we minimized the total training time by jointly optimizing the UAV deployment and integrated sensing, computation, and communication (ISCC) resources.
Then, we applied the alternating optimization technique and decomposed this challenging mixed-integer non-convex problem into three subproblems. 
Moreover, we proposed the BBPO scheme to alternately optimize bandwidth, batch size, and position, which efficiently computes suboptimal solutions.
Simulation results showed that our BBPO scheme outperforms other baselines in terms of convergence rate and testing accuracy.

\appendices
\section{Proof of Theorem 1}\label{app:Fbound}
Based on Assumption 1, the expected loss function $F\left(\mathbf w^{(n+1)}\right)$ can be expressed as
\begin{equation} \label{equ:F}
\begin{array}{rll}
\!\! \Exp{\!F\left(\mathbf w^{(n+1)}\right)\!}
\!\leq \!\Exp{\!F\left(\mathbf w^{(n)}\right)\!}
\!+\!\Exp{\!\left\langle\nabla F\left(\mathbf w^{(n)}\right), \mathbf w^{(n+1)}\!-\!\mathbf w^{(n)}\right\rangle\!}
\!+\!\displaystyle\frac L 2 \Exp{\!\left\Vert\mathbf w^{(n+1)}\!-\!\mathbf w^{(n)}\right\Vert^2\!}.
\end{array}
\end{equation}
We first provide two key lemmas to derive the convergence rate in Section \ref{sec:lemmas}. Then, the proofs of these lemmas are in Section \ref{sec:proof_lemmaFw} and \ref{sec:proof_lemmaww}.

\subsection{Convergence Analysis} \label{sec:lemmas}
%
\begin{lemma}\label{lemma:Fw}
With Assumption 3 and Assumption 4, it holds that
\begin{equation} \label{equ:Fw}
\begin{array}{rll}
\Exp{\left\langle\nabla F(\mathbf w^{(n)}), \mathbf w^{(n+1)}-\mathbf w^{(n)}\right\rangle}
\leq-\displaystyle\frac{\eta}{2}\Exp{\left\Vert\nabla F(\mathbf w^{(n)})\right\Vert^2}
+\eta\sum\limits_{k\in\mc K}\alpha_k^2\cdot\sum\limits_{k\in\mc K}\left(\displaystyle\frac{\sigma_k^2}{\delta_k}+\Lambda_k^2\right),
\end{array}
\end{equation}
where each $\alpha_k, \forall k\in\mc K$ is a function of $\left\{q_{s,k}(\bs u_k)\right\}$. 
Detailed expression of $\{\alpha_k\}$ is in \eqref{equ:g}.
\end{lemma}
\begin{lemma}\label{lemma:ww}
With Assumption 3 and Assumption 4, it holds that
\begin{equation} \label{equ:ww}
\begin{array}{rll}
\Exp{\left\Vert\mathbf w^{(n+1)}\!-\!\mathbf w^{(n)}\right\Vert^2}
\!\!\!\!&\leq2\eta^2\left(\!\sum\limits_{k\in\mc K}\!\beta_k\displaystyle\frac{\sigma_k^2}{\delta_k}
\!+\!2 \!\sum\limits_{k\in\mc K}\!\beta_k\Lambda_k^2
\!+\!2 \!\sum\limits_{k\in\mc K}\! \gamma_k\left(\left(q_{s,k}(\bs u_k)\!-\!\overline q_s\right)^2+\overline q_s^2\right)\Lambda_{k}^2\right)\\
&\!+\!4\eta^2 \Exp{\left\Vert\nabla F(\mathbf w^{(n)})\right\Vert^2},
\end{array}
\end{equation}
where each $\beta_k, \gamma_k, \forall k\in\mc K$ is a function of $\left\{q_{s,k}(\bs u_k)\right\}$. 
Detailed expressions for $\{\beta_k\}$ and $\{\gamma_k\}$ are contained in \eqref{equ:M11} and \eqref{equ:M22final} respectively.
\end{lemma}

By substituting \eqref{equ:Fw} and \eqref{equ:ww} into \eqref{equ:F}, we have
\begin{align}
&\Exp{F(\mathbf w^{(n+1)})} \nonumber\\
&\leq \Exp{F(\mathbf w^{(n)})}
-\left(\displaystyle\frac{\eta}{2}-2L\eta^2\right)\Exp{\left\Vert\nabla F(\mathbf w^{(n)})\right\Vert^2}
+\eta\sum\limits_{k\in\mc K}\alpha_k^2\cdot\sum\limits_{k\in\mc K}\left(\displaystyle\frac{\sigma_k^2}{\delta_k}+\Lambda_k^2\right) \notag\\
&+\displaystyle L\eta^2\left(\sum\limits_{k\in\mc K}\beta_k\displaystyle\frac{\sigma_k^2}{\delta_k}+2 \sum\limits_{k\in\mc K}\beta_k\Lambda_k^2+2 \sum\limits_{k\in\mc K} \gamma_k\left(\left(q_{s,k}(\bs u_k)-\overline q_s\right)^2+\overline q_s^2\right)\Lambda_{k}^2\right) \notag\\
&\leq \Exp{F(\mathbf w^{(n)})}
-\mu\eta\left(1-4L\eta\right)\Exp{F(\mathbf w^{(n)})-F(\mathbf w_*)}
+\eta\sum\limits_{k\in\mc K}\alpha_k^2 \sum\limits_{k\in\mc K}\left(\displaystyle\frac{\sigma_k^2}{\delta_k}+\Lambda_k^2\right) \notag\\
&+\displaystyle L\eta^2\left(\sum\limits_{k\in\mc K}\beta_k\displaystyle\frac{\sigma_k^2}{\delta_k}+2 \sum\limits_{k\in\mc K}\beta_k\Lambda_k^2+2 \sum\limits_{k\in\mc K} \gamma_k\left(\left(q_{s,k}-\overline q_s\right)^2+\overline q_s^2\right)\Lambda_{k}^2\right), \label{equ:F1}
\end{align}
where \eqref{equ:F1} is due to $\left\Vert\nabla F(\mathbf w^{(n)})\right\Vert^2\geq 2\mu\left(F(\mathbf w^{(n)})-F(\mathbf w_*)\right)$ based on Assumption 2.
Subtract $\Exp{F\left(\mathbf w_*\right)}$ in both sides of \eqref{equ:F1}, we have
\begin{align}
&\Exp{F(\mathbf w^{(n+1)})-F(\mathbf w_*)} \notag\\
&\leq
\left(1-\mu\eta\left(1-4L\eta\right)\right)\Exp{F\left(\mathbf w^{(n)}\right)-F\left(\mathbf w_*\right)} \notag\\
&\!+\!\underbrace{\eta\!\sum\limits_{k\in\mc K}\!\alpha_k^2 \!\sum\limits_{k\in\mc K}\!\left(\displaystyle\frac{\sigma_k^2}{\delta_k}\!+\!\Lambda_k^2\right)
\!+\!\displaystyle L\eta^2\!\sum\limits_{k\in\mc K}\!\beta_k\displaystyle\left(\frac{\sigma_k^2}{\delta_k}\!+\!2\Lambda_k^2\right)
\!+\!2L\eta^2 \!\sum\limits_{k\in\mc K}\! \gamma_k\left(\left(q_{s,k}(\bs u_k)\!-\!\overline q_s\right)^2\!+\!\overline q_s^2\right)\Lambda_{k}^2}_{\triangleq G}. \label{equ:F2}
\end{align}
Applying \eqref{equ:F2} recursively, we have
\begin{align}
\!\!\!\Exp{\!F\left(\mathbf w^{(n)}\right)\!-\!F(\mathbf w_*)\!}
\leq G\displaystyle\frac {1\!-\!\left(1\!-\!\mu\eta\left(1\!-\!4L\eta\right)\right)^n}{\mu\eta\left(1\!-\!4L\eta\right)}\!+\!\left(1\!-\!\mu\eta\left(1\!-\!4L\eta\right)\right)^n\Exp{\!F\left(\mathbf w^0\right)\!-\!F(\mathbf w_*)\!}. \label{equ:F3}
\end{align}
To guarantee the convergence of the learning process, we have $1-\mu\eta\left(1-4L\eta\right)<1$, which leads to $0<\eta<\frac 1 {4L}$.
Thus, Theorem 1 is proved.

\subsection{Proof of Lemma \ref{lemma:Fw}} \label{sec:proof_lemmaFw}
%
We first derive the desired expression for the aggregated gradient in each round $n$ as
\begin{align}
&\Exp{\displaystyle\frac{\sum\limits_{k\in\mc K} \mathbf 1_{k}^{(n)}(\bs u_k)\mathbf {g}_k^{(n)}} {\sum\limits_{k\in\mc K} \mathbf1_{k}^{(n)}(\bs u_k)}}{\sum\limits_{k\in\mc K}\mathbf1_{k}^{(n)}(\bs u_k)\neq 0} \notag\\
&=\Exp{\sum\limits_{l=1}^{K}\!\!\!\!\!\!\!\sum\limits_{\mc K_1\cup\mc K_2=\mc K, \atop |\mc K_1|=l, |\mc K_2|=K-l}\!\!\!\!\!\!\!\!\! \Pr\!\left(\!\!\mathbf 1_{k_1}^{(n)}(\bs u_{k_1})\!=\!1 \forall k_1\in\mc K_1, \mathbf 1_{k_2}^{(n)}(\bs u_{k_2})\!=\!0 \forall k_2\in\mc K_2 \middle| \sum\limits_{k\in\mc K}\mathbf 1_k^{(n)}(\bs u_k)\neq0\!\!\right)
\displaystyle\!\!\frac 1 l \!\!\sum\limits_{k_1\in\mc K_1} \!\!\mathbf g_{k_1}^{(n)}} \notag\\
&=\Exp{\sum\limits_{l=1}^{K}\sum\limits_{\mc K_1\cup\mc K_2=\mc K, \atop |\mc K_1|=l, |\mc K_2|=K-l}\displaystyle\frac{\prod\limits_{k_1\in\mc K_1}q_{s,k_1} \prod\limits_{k_2\in\mc K_2}(1-q_{s,k_2})}{1-\prod\limits_{k\in\mc K}(1-q_{s,k})}\cdot \displaystyle\frac 1 l \sum\limits_{k_1\in\mc K_1} \mathbf g_{k_1}^{(n)}} \notag\\
&\triangleq\sum\limits_{k\in\mc K}\alpha_k \Exp{\mathbf g_k^{(n)}}, \label{equ:g}
\end{align}
where $\mc K_1$ is the set of UAVs that successfully sense the target, and $\mc K_2$ is the set of UAVs that fail to sense the target, and do not participate in round $n$.  
Moreover, each $\alpha_k, \forall k\in\mc K$ is a function of $\left\{q_{s,k}(\bs u_k)\right\}$. Let $\mathbf g_k^{(n)}=1, \forall k\in\mc K$, we have $\sum\limits_{k\in\mc K}\alpha_k=1$.
For the convenience of expression, we use $q_{s,k_1}$ and $q_{s,k_2}$ to represent $q_{s,k_1}(\bs u_{k_1})$ and $q_{s,k_1}(\bs u_{k_2})$ respectively. 

Next, we proof Lemma \ref{lemma:Fw} as follows
\begin{align}
&\Exp{\left\langle\nabla F(\mathbf w^{(n)}), \mathbf w^{(n+1)}-\mathbf w^{(n)}\right\rangle}\label{equ:proofFw}\\
&=-\eta\Exp{\left\langle\nabla F(\mathbf w^{(n)}),\sum\limits_{k\in\mc K}\alpha_k \mathbf g_k^{(n)}\right\rangle} \tag{\ref{equ:proofFw}{a}} \label{equ:proofFwa}\\
&=\displaystyle\frac{\eta}{2}\Exp{
\left\Vert\nabla F(\mathbf w^{(n)})-\sum\limits_{k\in\mc K}\alpha_k \mathbf g_k^{(n)}\right\Vert^2
-\left\Vert\nabla F(\mathbf w^{(n)})\right\Vert^2
-\left\Vert\sum\limits_{k\in\mc K}\alpha_k \mathbf g_k^{(n)}\right\Vert^2}
\tag{\ref{equ:proofFw}{b}} \label{equ:proofFwb}\\
&\leq\displaystyle\frac{\eta}{2}\Exp{
\left\Vert\nabla F(\mathbf w^{(n)})-\sum\limits_{k\in\mc K}\alpha_k \mathbf g_k^{(n)}\right\Vert^2
-\left\Vert\nabla F(\mathbf w^{(n)})\right\Vert^2}
\notag\\
&=\displaystyle\frac{\eta}{2}\Exp{\left\Vert\nabla F(\mathbf w^{(n)})\!-\!\sum\limits_{k\in\mc K}\alpha_k F_k(\mathbf w^{(n)})\!+\!\sum\limits_{k\in\mc K}\alpha_k F_k(\mathbf w^{(n)})\!-\!\sum\limits_{k\in\mc K}\alpha_k \mathbf g_k^{(n)}\right\Vert^2}
\!\!-\!\displaystyle\frac{\eta}{2}\Exp{\left\Vert\nabla F(\mathbf w^{(n)})\right\Vert^2}
\notag\\
&\leq\eta\Exp{\!\left\Vert\!\sum\limits_{k\in\mc K}\!\alpha_k \!\left(\nabla F\!(\mathbf w^{(n)}\!)\!-\!F_k\!(\mathbf w^{(n)}\!)\!\right)\right\Vert^2\!\!\!+\!\left\Vert\!\sum\limits_{k\in\mc K}\!\alpha_k \!\left(F_k\!(\mathbf w^{(n)}\!)\!-\!\mathbf g_k^{(n)}\!\right)\right\Vert^2\!}
\!\!-\! \frac{\eta}{2}\Exp{\!\left\Vert\nabla F\!(\mathbf w^{(n)}\!)\right\Vert^2\!}
\tag{\ref{equ:proofFw}{c}} \label{equ:proofFwc}\\
&\leq\eta\!\sum\limits_{k\in\mc K}\!\alpha_k^2 \Exp{\!\sum\limits_{k\in\mc K}\!\left\Vert\nabla F\!(\mathbf w^{(n)}\!)\!-\!F_k\!(\mathbf w^{(n)}\!)\right\Vert^2
\!+\!\sum\limits_{k\in\mc K}\!\left\Vert F_k\!(\mathbf w^{(n)}\!)\!-\!\mathbf g_k^{(n)}\right\Vert^2\!}
\!-\!\frac{\eta}{2}\Exp{\!\left\Vert\nabla F\!(\mathbf w^{(n)}\!)\right\Vert^2\!}
\tag{\ref{equ:proofFw}{d}} \label{equ:proofFwd}\\
&\leq\eta\sum\limits_{k\in\mc K}\alpha_k^2\cdot\sum\limits_{k\in\mc K}\left(\Lambda_k^2
+\displaystyle\frac{\sigma_k^2}{\delta_k}\right)
-\displaystyle\frac{\eta}{2}\Exp{\left\Vert\nabla F(\mathbf w^{(n)})\right\Vert^2},
\tag{\ref{equ:proofFw}{e}} \label{equ:proofFwe}
\end{align}
where \eqref{equ:proofFwa} is obtained by \eqref{equ:g},
 \eqref{equ:proofFwb} follows from the basic identity $\langle \!\mathbf a,\!\mathbf b\!\rangle\!\!=\!\!\frac 1 2\!\left(\Vert \mathbf a\Vert^2\!\!+\!\!\Vert \mathbf b\Vert^2\!\!-\!\!\Vert \mathbf a\!\!-\!\!\mathbf b\Vert^2\!\right)$,
\eqref{equ:proofFwc} is due to $\Vert \mathbf a+\mathbf b\Vert^2\leq2\left(\Vert \mathbf a\Vert^2+\Vert\mathbf b\Vert^2\right)$ and $\nabla F(\mathbf w^{(n)})=\sum\limits_{k\in\mc K}\left(\alpha_k\nabla F(\mathbf w^{(n)})\right)$,
\eqref{equ:proofFwd} is due to the Cauchy-Schwarz Inequality $\left\Vert\sum\limits_{k\in\mc K}\mathbf a_k \mathbf b_k\right\Vert^2\leq\sum\limits_{k\in\mc K}\Vert\mathbf a_k\Vert^2+\sum\limits_{k\in\mc K}\Vert\mathbf b_k\Vert^2$, 
and inequality \eqref{equ:proofFwe} is due to Assumption 3 and Assumption 4.
Therefore, we obtain Lemma \ref{lemma:Fw}.

\subsection{Proof of Lemma \ref{lemma:ww}} \label{sec:proof_lemmaww}
%
\begin{align}
&\Exp{\left\Vert\mathbf w^{(n+1)}-\mathbf w^{(n)}\right\Vert^2} \label{equ:proofww}\\
&=\eta^2\Exp{\left\Vert\displaystyle\frac{\sum\limits_{k\in\mc K} \mathbf 1_{k}^{(n)}(\bs u_k)\left(\mathbf {g}_k^{(n)}-\nabla F_k(\mathbf w^{(n)})\right)}{\sum\limits_{k\in\mc K} \mathbf1_{k}^{(n)}(\bs u_k)}
+\displaystyle\frac{\sum\limits_{k\in\mc K} \mathbf 1_{k}^{(n)}(\bs u_k)\nabla F_k(\mathbf w^{(n)})} {\sum\limits_{k\in\mc K} \mathbf1_{k}^{(n)}(\bs u_k)}\right\Vert^2}
\notag\\
&\leq2\eta^2\underbrace{\Exp{\left\Vert\displaystyle\frac{\sum\limits_{k\in\mc K} \mathbf 1_{k}^{(n)}(\bs u_k)\left(\mathbf {g}_k^{(n)}-\nabla F_k(\mathbf w^{(n)})\right)}{\sum\limits_{k\in\mc K} \mathbf1_{k}^{(n)}(\bs u_k)}\right\Vert^2}}_{\triangleq M_1}
\!+\!2\eta^2\underbrace{\Exp{\left\Vert\displaystyle\frac{\sum\limits_{k\in\mc K} \mathbf 1_{k}^{(n)}(\bs u_k)\nabla F_k(\mathbf w^{(n)})} {\sum\limits_{k\in\mc K} \mathbf1_{k}^{(n)}(\bs u_k)}\right\Vert^2}}_{\triangleq M_2},
\tag{\ref{equ:proofww}{a}} \label{equ:proofwwa}
\end{align}
where \eqref{equ:proofwwa} is due to $\Vert \mathbf a+\mathbf b\Vert^2\leq2\left(\Vert \mathbf a\Vert^2+\Vert\mathbf b\Vert^2\right)$.
Next we bound $M_1$ and $M_2$ as follows.

\subsubsection{Bound of $M_1$}
\begin{align}
 M_1 
&=\Exp{\frac{\sum\limits_{k\in\mc K} \mathbf 1_{k}^{(n)}(\bs u_k)\left\Vert\mathbf {g}_k^{(n)} - \nabla F_k(\mathbf w^{(n)})\right\Vert^2}{\left(\sum\limits_{k\in\mc K} \mathbf1_{k}^{(n)}(\bs u_k)\right)^2}} \notag\\
&+\underbrace{\Exp{\displaystyle\frac{\sum\limits_{i\in\mc K}\sum\limits_{j\in\mc K,\atop j\neq i} \mathbf 1_{i}^{(n)}(\bs u_i) \mathbf 1_{j}^{(n)}(\bs u_j)\left(\mathbf {g}_i^{(n)}-\nabla F_i(\mathbf w^{(n)})\right)\left(\mathbf {g}_j^{(n)}-\nabla F_j(\mathbf w^{(n)})\right)}{\left(\sum\limits_{k\in\mc K} \mathbf1_{k}^{(n)}(\bs u_k)\right)^2}}}_{=0}. \label{equ:M1}
\end{align}
Since we have $\mathbb E \left[\mathbf g_k^{(n)}\right]=\nabla F_k\left(\mathbf w^{(n)}\right), \forall k\in\mc K$, the last term in \eqref{equ:M1} is $0$.
Thus,
\begin{align}
&M_{1} \nonumber\\
&=\!\Exp{\!\sum\limits_{l=1}^{K}\!\!\!\!\!\!\!\!\sum\limits_{\mc K_1\!\cup\mc K_2=\mc K, \atop |\mc K_1|=l, |\mc K_2|=K-l}\!\!\!\!\!\!\!\!\!\!\!\! \Pr\!\!\left(\!\!\mathbf 1_{k_1}^{(\!n\!)}\!(\!\bs u_{k_1}\!)\!\!=\!1 \forall k_1\!\!\in\!\!\mc K_1,\! \mathbf 1_{k_2}^{(\!n\!)}\!(\!\bs u_{k_2}\!)\!\!=\!0 \forall k_2\!\!\in\!\!\mc K_2 \middle| \sum\limits_{k\in\mc K}\!\!\mathbf 1_k^{(\!n\!)}\!(\!\bs u_k\!)\!\neq\!0\!\!\right)
\!\!\!\frac 1 {l^2 }\!\!\!\sum\limits_{k_1\in\mc K_1}\!\!\! \left\Vert\mathbf {g}_k^{(n)}\!\!\!-\!\!\nabla F_k(\!\mathbf w^{(n)}\!)\right\Vert^2\!} 
\notag\\
&=\Exp{\sum\limits_{l=1}^{K}\sum\limits_{\mc K_1\cup\mc K_2=\mc K, \atop |\mc K_1|=l, |\mc K_2|=K-l}\displaystyle\frac{\prod\limits_{k_1\in\mc K_1}q_{s,k_1} \prod\limits_{k_2\in\mc K_2}(1-q_{s,k_2})}{1-\prod\limits_{k\in\mc K}(1-q_{s,k})}\cdot \displaystyle\frac 1 {l^2} \sum\limits_{k_1\in\mc K_1} \left\Vert\mathbf {g}_k^{(n)}-\nabla F_k(\mathbf w^{(n)})\right\Vert^2}
\notag\\
&\triangleq\sum\limits_{k\in\mc K}\beta_k \Exp{\left\Vert\mathbf {g}_k^{(n)}-\nabla F_k(\mathbf w^{(n)})\right\Vert^2} \notag\\
&\leq\sum\limits_{k\in\mc K}\beta_k\displaystyle\frac{\sigma_k^2}{\delta_k}, \label{equ:M11} 
\end{align}
where each $\beta_k, \forall k\in\mc K$ is a function of $\{q_{s,k}\}$, 
and \eqref{equ:M11} is due to $\mathbb E\left[\left\Vert \mathbf g_k^{(n)}-\nabla F_k(\mathbf w^{(n)})\right\Vert^2\right]\leq\frac{\sigma_k^2}{\delta_k}$.

\subsubsection{Bound of $M_2$}
\begin{align}
M_2
&=\Exp{\displaystyle\frac{\left\Vert\sum\limits_{k\in\mc K} \mathbf 1_{k}^{(n)}(\bs u_k)\left(\nabla F_k(\mathbf w^{(n)})-\nabla F(\mathbf w^{(n)})\right)+\sum\limits_{k\in\mc K} \mathbf 1_{k}^{(n)}(\bs u_k)\nabla F(\mathbf w^{(n)})\right\Vert^2}{\left(\sum\limits_{k\in\mc K} \mathbf1_{k}^{(n)}(\bs u_k)\right)^2}}\notag \\
&\leq2\Exp{\displaystyle\frac{\left\Vert\sum\limits_{k\in\mc K} \mathbf 1_{k}^{(n)}(\bs u_k)\left(\nabla F_k(\mathbf w^{(n)})-\nabla F(\mathbf w^{(n)})\right)\right\Vert^2}{\left(\sum\limits_{k\in\mc K} \mathbf1_{k}^{(n)}(\bs u_k)\right)^2}+\frac{\left\Vert\sum\limits_{k\in\mc K} \mathbf 1_{k}^{(n)}(\bs u_k)\nabla F(\mathbf w^{(n)})\right\Vert^2}{\left(\sum\limits_{k\in\mc K} \mathbf1_{k}^{(n)}(\bs u_k)\right)^2}} \notag\\
&=2\Exp{\displaystyle\frac{\left\Vert\sum\limits_{k\in\mc K} \mathbf 1_{k}^{(n)}(\bs u_k)\left(\nabla F_k(\mathbf w^{(n)})\!-\!\nabla F(\mathbf w^{(n)})\right)\right\Vert^2}{\left(\sum\limits_{k\in\mc K} \mathbf1_{k}^{(n)}(\bs u_k)\right)^2}}
\!\!+\!2\Exp{\frac{\left(\!\sum\limits_{k\in\mc K} \mathbf 1_{k}^{(n)}(\bs u_k)\!\!\right)^2 \left\Vert\nabla F(\mathbf w^{(n)})\right\Vert^2}{\left(\sum\limits_{k\in\mc K} \mathbf1_{k}^{(n)}(\bs u_k)\right)^2}} \notag\\
&=2\underbrace{\Exp{\displaystyle\frac{\sum\limits_{k\in\mc K} \mathbf 1_{k}^{(n)}(\bs u_k)\left\Vert\nabla F_k(\mathbf w^{(n)})-\nabla F(\mathbf w^{(n)})\right\Vert^2}{\left(\sum\limits_{k\in\mc K} \mathbf1_{k}^{(n)}(\bs u_k)\right)^2}} }_{\triangleq M_{21}}
+2\Exp{\left\Vert\nabla F(\mathbf w^{(n)})\right\Vert^2} \notag\\
&+2\underbrace{\Exp{\displaystyle\frac{\sum\limits_{i\in\mc K}\sum\limits_{j\in\mc K,\atop j\neq i} \mathbf 1_{i}^{(n)}(\bs u_i) \mathbf 1_{j}^{(n)}(\bs u_j)\!\!\left(\nabla F_i(\mathbf w^{(n)})\!-\!\nabla F(\mathbf w^{(n)})\right)\!\!\left(\nabla F_j(\mathbf w^{(n)})\!-\!\nabla F(\mathbf w^{(n)})\right)}{\left(\sum\limits_{k\in\mc K} \mathbf1_{k}^{(n)}(\bs u_k)\right)^2}}}_{\triangleq M_{22}}. \label{equ:M2}
\end{align}
Next, we bound $M_{21}$ and $M_{22}$ in \eqref{equ:M2} as follows.
Firstly, similar to \eqref{equ:M11}, we have
\begin{align}
M_{21}
&\triangleq\sum\limits_{k\in\mc K}\beta_k \Exp{\left\Vert\nabla F_k(\mathbf w^{(n)})-\nabla F(\mathbf w^{(n)})\right\Vert^2} \notag\\
&\leq\sum\limits_{k\in\mc K}\beta_k\Lambda_k^2, \label{equ:M21}
\end{align}
where \eqref{equ:M21} is due to $\mathbb E\left[\left\Vert\nabla F_k(\mathbf w^{(n)})-\nabla F(\mathbf w^{(n)})\right\Vert^2\right] \leq\Lambda_k^2$ in Assumption 4.
Secondly,
\begin{align}
M_{22}
&=\mathbb E\Bigg[\sum\limits_{l=2}^{K}\!\!\!\!\!\!\sum\limits_{\mc K_1\cup\mc K_2=\mc K, \atop |\mc K_1|=l, |\mc K_2|=K-l} \!\!\!\!\!\!\!\!\!\Pr\left(\mathbf 1_{k_1}^{(n)}(\bs u_{k_1})=1 \forall k_1\in\mc K_1, \mathbf 1_{k_2}^{(n)}(\bs u_{k_2})=0 \forall k_2\in\mc K_2 \middle| \sum\limits_{k\in\mc K}\mathbf 1_k^{(n)}(\bs u_{k})\!\geq\! 2\right) \notag\\
&\cdot \displaystyle\frac 1 {l^2} \sum\limits_{k_1\in\mc K_1}\sum\limits_{k'_1\in\mc K_1, k'_1\neq k_1} \left(\nabla F_{k_1}(\mathbf w^{(n)})-\nabla F(\mathbf w^{(n)})\right)
\left(\nabla F_{k'_1}(\mathbf w^{(n)})-\nabla F(\mathbf w^{(n)})\right)\Bigg] \notag\\
&=\mathbb E\Bigg[\sum\limits_{l=2}^{K}\sum\limits_{\mc K_1\cup\mc K_2=\mc K, \atop |\mc K_1|=l, |\mc K_2|=K-l}\displaystyle\frac{\prod\limits_{k_1\in\mc K_1}q_{s,k_1} \prod\limits_{k_2\in\mc K_2}(1-q_{s,k_2})}{1-\prod\limits_{k\in\mc K}(1-q_{s,k})-\sum\limits_{k\in\mc K}q_{s,k}\prod\limits_{k'\in\mc K, \atop k'\neq k}(1-q_{s,k'})} \notag\\
&\cdot \displaystyle\frac 1 {l^2} \sum\limits_{k_1\in\mc K_1}\sum\limits_{k'_1\in\mc K_1, k'_1\neq k_1} \left(\nabla F_{k_1}(\mathbf w^{(n)})-\nabla F(\mathbf w^{(n)})\right)
\left(\nabla F_{k'_1}(\mathbf w^{(n)})-\nabla F(\mathbf w^{(n)})\right)\Bigg] \notag\\
&\leq\sum\limits_{l=2}^{K}\sum\limits_{\mc K_1\cup\mc K_2=\mc K, \atop |\mc K_1|=l, |\mc K_2|=K-l}\displaystyle\frac{\prod\limits_{k_2\in\mc K_2}(1-q_{s,k_2})}{1-\prod\limits_{k\in\mc K}(1-q_{s,k})-\sum\limits_{k\in\mc K}q_{s,k}\prod\limits_{k'\in\mc K, \atop k'\neq k}(1-q_{s,k'})} \notag\\
&\cdot \frac 1 {l^2} \underbrace{\Exp{\!\sum\limits_{k_1\in\mc K_1}\!\sum\limits_{k'_1\in\mc K_1, \atop k'_1\neq k_1}\! q_{s,k_1}\left(\nabla F_{k_1}(\mathbf w^{(n)})-\nabla F(\mathbf w^{(n)})\right)\cdot q_{s,k'_1}\left(\nabla F_{k'_1}(\mathbf w^{(n)})-\nabla F(\mathbf w^{(n)})\right)\!}}_{\triangleq M_{221}} \label{equ:M22}
\end{align}
where \eqref{equ:M22} is due to $\prod\limits_{k_1\in\mc K_1}q_{s,k_1}\leq q_{s,k_1}q_{s,k'_1}, \forall k,k'\in\mc K_1, k'\neq k$.
Next, we define the average successful sensing probability as $\overline q=\frac 1 K\sum\limits_{k\in\mc K}q_{s,k}$, and we have
\begin{align}
&M_{221} \label{equ:M221}\\
&=\! \mathbb E\Bigg[\!\sum\limits_{k_1\in\mc K_1}\!\!\sum\limits_{k'_1\in\mc K_1, \atop k'_1\neq k_1} \!\!\!\! \left(q_{s,k_1}\!-\!\overline q_s\!+\!\overline q_s\right)\!\left(\nabla F_{k_1}(\!\mathbf w^{(n)}\!)\!-\!\!\nabla F(\!\mathbf w^{(n)}\!)\right)
\!\cdot \!\left(q_{s,k'_1}\!-\!\overline q_s\!+\!\overline q_s\right)\!\left(\nabla F_{k'_1}(\!\mathbf w^{(n)}\!)\!-\!\!\nabla F(\!\mathbf w^{(n)}\!)\right)\!\Bigg] \notag\\
&=\mathbb E\Bigg[\!\sum\limits_{k_1\in\mc K_1}\!\sum\limits_{k'_1\in\mc K_1, k'_1\neq k_1}\!\!\!\!\!\!\!\Big[\! \left(q_{s,k_1}\!-\!\overline q_s\right)\!\left(\nabla F_{k_1}(\mathbf w^{(n)})\!-\!\nabla F(\mathbf w^{(n)})\right)
\!\cdot\! \left(q_{s,k'_1}\!-\!\overline q_s\right)\!\left(\nabla F_{k'_1}(\mathbf w^{(n)})\!-\!\nabla F(\mathbf w^{(n)})\right) \notag\\
&+\left(q_{s,k_1}-\overline q_s\right)\left(\nabla F_{k_1}(\mathbf w^{(n)})-\nabla F(\mathbf w^{(n)})\right)\cdot \overline q_s\left(\nabla F_{k'_1}(\mathbf w^{(n)})-\nabla F(\mathbf w^{(n)})\right) \notag\\
&+\overline q_s\left(\nabla F_{k_1}(\mathbf w^{(n)})-\nabla F(\mathbf w^{(n)})\right)\cdot \left(q_{s,k'_1}-\overline q_s\right)\left(\nabla F_{k'_1}(\mathbf w^{(n)})-\nabla F(\mathbf w^{(n)})\right) \notag\\
&+\overline q_s\left(\nabla F_{k_1}(\mathbf w^{(n)})-\nabla F(\mathbf w^{(n)})\right)\cdot \overline q_s\left(\nabla F_{k'_1}(\mathbf w^{(n)})-\nabla F(\mathbf w^{(n)})\right) \!\Big]\!\Bigg] \notag\\
&\leq\frac 1 2\mathbb E\Bigg[\!\sum\limits_{k_1\in\mc K_1}\!\!\sum\limits_{k'_1\in\mc K_1, \atop k'_1\neq k_1}\!\!\!\Big[\!
\left(q_{s,k_1}\!-\!\overline q_s\right)^2\!\left\Vert\nabla F_{k_1}(\!\mathbf w^{(n)}\!)\!-\!\nabla F(\!\mathbf w^{(n)}\!)\right\Vert^2 \!\!\!+\! \left(q_{s,k'_1}\!-\!\overline q_s\right)^2\!\left\Vert\nabla F_{k'_1}(\!\mathbf w^{(n)}\!)\!-\!\nabla F(\!\mathbf w^{(n)}\!)\right\Vert^2 \notag\\
&+\left(q_{s,k_1}-\overline q_s\right)^2\left\Vert\nabla F_{k_1}(\mathbf w^{(n)})-\nabla F(\mathbf w^{(n)})\right\Vert^2 + \overline q_s^2\left\Vert\nabla F_{k'_1}(\mathbf w^{(n)})-\nabla F(\mathbf w^{(n)})\right\Vert^2 \notag\\
&+\overline q_s^2\left\Vert\nabla F_{k_1}(\mathbf w^{(n)})-\nabla F(\mathbf w^{(n)})\right\Vert^2 + \left(q_{s,k'_1}-\overline q_s\right)^2\left\Vert\nabla F_{k'_1}(\mathbf w^{(n)})-\nabla F(\mathbf w^{(n)})\right\Vert^2 \notag\\
&+\overline q_s^2\left\Vert\nabla F_{k_1}(\mathbf w^{(n)})-\nabla F(\mathbf w^{(n)})\right\Vert^2 + \overline q_s^2\left\Vert\nabla F_{k'_1}(\mathbf w^{(n)})-\nabla F(\mathbf w^{(n)})\right\Vert^2 \Big]\!\Bigg] \tag{\ref{equ:M221}{a}} \label{equ:M221a}\\
&=\sum\limits_{k_1\in\mc K_1}\sum\limits_{k'_1\in\mc K_1, k'_1\neq k_1}\Bigg(
\left(\left(q_{s,k_1}-\overline q_s\right)^2+\overline q_s^2\right)
\Exp{\left\Vert\nabla F_{k_1}(\mathbf w^{(n)})-\nabla F(\mathbf w^{(n)})\right\Vert^2} \notag\\ 
&+ \left(\left(q_{s,k'_1}-\overline q_s\right)^2+\overline q_s^2\right)\Exp{\left\Vert\nabla F_{k'_1}(\mathbf w^{(n)})-\nabla F(\mathbf w^{(n)})\right\Vert^2}\Bigg) \notag\\
&=\sum\limits_{k_1\in\mc K_1}\Bigg(
(l-1)\left(\left(q_{s,k_1}-\overline q_s\right)^2+\overline q_s^2\right)
\Exp{\left\Vert\nabla F_{k_1}(\mathbf w^{(n)})-\nabla F(\mathbf w^{(n)})\right\Vert^2} \notag\\ 
&+ \sum\limits_{k'_1\in\mc K_1, k'_1\neq k_1}\left(\left(q_{s,k'_1}-\overline q_s\right)^2+\overline q_s^2\right)\Exp{\left\Vert\nabla F_{k'_1}(\mathbf w^{(n)})-\nabla F(\mathbf w^{(n)})\right\Vert^2}\Bigg) \tag{\ref{equ:M221}{b}} \label{equ:M221b}\\
&=\sum\limits_{k_1\in\mc K_1}\Bigg(
(l-2)\left(\left(q_{s,k_1}-\overline q_s\right)^2+\overline q_s^2\right)
\Exp{\left\Vert\nabla F_{k_1}(\mathbf w^{(n)})-\nabla F(\mathbf w^{(n)})\right\Vert^2} \notag\\ 
&+ \sum\limits_{k_1\in\mc K_1}\left(\left(q_{s,k_1}-\overline q_s\right)^2+\overline q_s^2\right)\Exp{\left\Vert\nabla F_{k_1}(\mathbf w^{(n)})-\nabla F(\mathbf w^{(n)})\right\Vert^2}\Bigg) \notag\\
&=2(l-1)\sum\limits_{k_1\in\mc K_1} 
\left(\left(q_{s,k_1}-\overline q_s\right)^2+\overline q_s^2\right)
\Exp{\left\Vert\nabla F_{k_1}(\mathbf w^{(n)})-\nabla F(\mathbf w^{(n)})\right\Vert^2} \notag\\
&\leq2(l-1)\sum\limits_{k_1\in\mc K_1} 
\left(\left(q_{s,k_1}-\overline q_s\right)^2+\overline q_s^2\right)\Lambda_{k_1}^2 
\tag{\ref{equ:M221}{c}} \label{equ:M221c}
\end{align}
where \eqref{equ:M221a} is due to $\langle \mathbf a, \mathbf b\rangle\leq \frac 1 2 (\Vert \mathbf a\Vert^2+\Vert \mathbf b\Vert^2)$,
\eqref{equ:M221b} is due to $|\mc K_1|=l$,
\eqref{equ:M221c} is due to $\mathbb E\left[\left\Vert\nabla F_k(\mathbf w^{(n)})-\nabla F(\mathbf w^{(n)})\right\Vert^2\right]\leq\Lambda_k^2$ in Assumption 4.
Substituting \eqref{equ:M221c} into \eqref{equ:M22}, we have
\begin{align}
M_{22}
&\leq\sum\limits_{l=2}^{K}\!\! \sum\limits_{\mc K_1\cup\mc K_2=\mc K, \atop |\mc K_1|=l, |\mc K_2|=K-l}\!\!\!\!\!\! \frac{\prod\limits_{k_2\in\mc K_2}\!\!(1\!-\!q_{s,k_2})}{1\!-\!\prod\limits_{k\in\mc K}\!\!(1-q_{s,k})\!-\!\sum\limits_{k\in\mc K}\!\!q_{s,k}\!\!\prod\limits_{k'\in\mc K, \atop k'\neq k}(1\!-\!q_{s,k'})} 
\!\cdot\! \frac {2(l\!-\!1)} {l^2} \!\sum\limits_{k_1\in\mc K_1}\!\!
\left(\left(q_{s,k_1}\!-\!\overline q_s\right)^2\!+\!\overline q_s^2\right)\Lambda_{k_1}^2\Bigg] \notag\\
&<\sum\limits_{l=2}^{K}\!\! \sum\limits_{\mc K_1\cup\mc K_2=\mc K, \atop |\mc K_1|=l, |\mc K_2|=K-l}\!\!\!\!\!\! \frac{\prod\limits_{k_2\in\mc K_2}\!\!(1\!-\!q_{s,k_2})}{1\!-\!\prod\limits_{k\in\mc K}(1\!-\!q_{s,k})\!-\!\sum\limits_{k\in\mc K}\!q_{s,k}\!\!\prod\limits_{k'\in\mc K, \atop k'\neq k}(1\!-\!q_{s,k'})} 
\!\cdot\! \displaystyle\frac {2} {l} \sum\limits_{k_1\in\mc K_1} \left(\left(q_{s,k_1}\!-\!\overline q_s\right)^2\!+\!\overline q_s^2\right)\Lambda_{k_1}^2\Bigg] \notag\\
&\triangleq\sum\limits_{k\in\mc K} \gamma_k\left(\left(q_{s,k}-\overline q_s\right)^2+\overline q_s^2\right)\Lambda_{k}^2, \label{equ:M22final}
\end{align}
where each $\gamma_k, \forall k\in\mc K$ is a function of $\{q_{s,k}\}$.

By substituting \eqref{equ:M21} and \eqref{equ:M22final} into \eqref{equ:M2}, we obtain the bound of $M_2$ as follows.
\begin{align}
M_2
&<2\sum\limits_{k\in\mc K}\beta_k\Lambda_k^2+2\sum\limits_{k\in\mc K} \gamma_k\left(\left(q_{s,k}(\bs u_k)-\overline q_s\right)^2+\overline q_s^2\right)\Lambda_{k}^2+2\Exp{\left\Vert\nabla F(\mathbf w^{(n)})\right\Vert^2} \label{equ:M2final}
\end{align}
Finally, by substituting \eqref{equ:M1} and \eqref{equ:M2final} into \eqref{equ:proofww},
we obtain Lemma \ref{lemma:ww}.
\begin{align}
&\Exp{\left\Vert\mathbf w^{(n+1)}-\mathbf w^{(n)}\right\Vert^2} \notag\\
&\leq2\eta^2\left(\sum\limits_{k\in\mc K}\beta_k\displaystyle\frac{\sigma_k^2}{\delta_k}
\!+\!2 \sum\limits_{k\in\mc K}\beta_k\Lambda_k^2
\!+\!2 \sum\limits_{k\in\mc K} \gamma_k\left(\left(q_{s,k}\!-\!\overline q_s\right)^2\!+\!\overline q_s^2\right)\Lambda_{k}^2\right)
\!+\!4\eta^2 \Exp{\left\Vert\nabla F(\mathbf w^{(n)})\right\Vert^2}. \label{equ:proofwwfinal}
\end{align}
%

\section{Proof of Corollary 1}\label{app:Gsameq}
We consider that all UAVs have a uniform successful sensing probability $q_{s,k}(\bs u_k)=q_s, \forall k\in\mc K$.
We first explore the relationship between $\alpha_k, \beta_k, \gamma_k, \forall k\in\mc K$ and $q_s$.
Then, we derive an upper bound on $G$.

\subsection{Expression of $\{\alpha_k\}$}
%
Due to $q_{s,k}(\bs u_k)=q_s, \forall k\in\mc K$, we can simplify \eqref{equ:g} to 
\begin{align}
&\Exp{\sum\limits_{l=1}^{K}\sum\limits_{\mc K_1\cup\mc K_2=\mc K, \atop |\mc K_1|=l, |\mc K_2|=K-l}\displaystyle\frac{\prod\limits_{k_1\in\mc K_1}q_{s,k_1} \prod\limits_{k_2\in\mc K_2}(1-q_{s,k_2})}{1-\prod\limits_{k\in\mc K}(1-q_{s,k})}\cdot \displaystyle\frac 1 l \sum\limits_{k_1\in\mc K_1} \mathbf g_{k_1}^{(n)}} \label{equ:alpha}\\
&=\Exp{\sum\limits_{l=1}^{K}\displaystyle\frac{q_{s}^l (1-q_{s})^{K-l}}{1-(1-q_{s})^K}\cdot \displaystyle\frac 1 l \cdot\sum\limits_{\mc K_1\cup\mc K_2=\mc K, \atop |\mc K_1|=l, |\mc K_2|=K-l}\sum\limits_{k_1\in\mc K_1} \mathbf g_{k_1}^{(n)}} \notag\\
&=\Exp{\sum\limits_{l=1}^{K}\displaystyle\frac{q_{s}^l (1-q_{s})^{K-l}}{1-(1-q_{s})^K}\cdot \displaystyle\frac 1 l \cdot C_{K-1}^{l-1}\sum\limits_{k\in\mc K} \mathbf g_{k}^{(n)}} 
 \tag{\ref{equ:alpha}{a}} \label{equ:alphaa}\\
&=\Exp{\sum\limits_{l=1}^{K} C_{K}^{l} \displaystyle\frac{q_{s}^l (1-q_{s})^{K-l}}{1-(1-q_{s})^K}\cdot \displaystyle\frac 1 K \cdot \sum\limits_{k\in\mc K} \mathbf g_{k}^{(n)}} 
 \tag{\ref{equ:alpha}{b}} \label{equ:alphab}\\
&=\frac 1 K\sum\limits_{k\in\mc K} \Exp{\mathbf g_k^{(n)}} \tag{\ref{equ:alpha}{c}} \label{equ:alphac}\\
&\triangleq\sum\limits_{k\in\mc K}\alpha_k \Exp{\mathbf g_k^{(n)}}, \notag
\end{align}
where \eqref{equ:alphaa} is due to $\sum\limits_{\mc K_1\cup\mc K_2=\mc K, \atop |\mc K_1|=l, |\mc K_2|=K-l}\sum\limits_{k_1\in\mc K_1} \mathbf g_{k_1}^{(n)}= C_{K-1}^{l-1}\sum\limits_{k\in\mc K} \mathbf g_{k}^{(n)}$,
\eqref{equ:alphab} is due to $\frac K l C_{K-1}^{l-1}=C_K^l$,
and \eqref{equ:alphac} is due to $\sum\limits_{l=1}^{K} C_{K}^{l} \frac{q_{s}^l (1-q_{s})^{K-l}}{1-(1-q_{s})^K}=1$. Therefore, we have $\alpha_k=\frac 1 K, \forall k\in\mc K$.

\subsection{Expression of $\{\beta_k\}$}
Due to $q_{s,k}(\bs u_k)=q_s, \forall k\in\mc K$, we can simplify \eqref{equ:M11} to 
\begin{align}
&\Exp{\sum\limits_{l=1}^{K}\sum\limits_{\mc K_1\cup\mc K_2=\mc K, \atop |\mc K_1|=l, |\mc K_2|=K-l}\displaystyle\frac{\prod\limits_{k_1\in\mc K_1}q_{s,k_1} \prod\limits_{k_2\in\mc K_2}(1-q_{s,k_2})}{1-\prod\limits_{k\in\mc K}(1-q_{s,k})}\cdot \displaystyle\frac 1 {l^2} \sum\limits_{k_1\in\mc K_1} \left\Vert\mathbf {g}_k^{(n)}-\nabla F_k(\mathbf w^{(n)})\right\Vert^2}
\label{equ:beta}\\
&=\Exp{\sum\limits_{l=1}^{K}\displaystyle\frac{q_{s}^l (1-q_{s})^{K-l}}{1-(1-q_{s})^K}\cdot \displaystyle\frac 1 {l^2} \cdot \sum\limits_{\mc K_1\cup\mc K_2=\mc K, \atop |\mc K_1|=l, |\mc K_2|=K-l}\sum\limits_{k_1\in\mc K_1} \left\Vert\mathbf {g}_k^{(n)}-\nabla F_k(\mathbf w^{(n)})\right\Vert^2}
\notag\\
&=\Exp{\sum\limits_{l=1}^{K}\displaystyle\frac {C_K^l }{l}\cdot\displaystyle\frac{q_{s}^l (1-q_{s})^{K-l}}{1-(1-q_{s})^K}\cdot \frac 1 K \sum\limits_{k\in\mc K}\left\Vert\mathbf {g}_k^{(n)}-\nabla F_k(\mathbf w^{(n)})\right\Vert^2}
\tag{\ref{equ:beta}{a}} \label{equ:betaa}\\
&\leq\frac{2}{K^2 \chi q_s}\sum\limits_{k\in\mc K}\Exp{\left\Vert\mathbf {g}_k^{(n)}-\nabla F_k(\mathbf w^{(n)})\right\Vert^2}
\tag{\ref{equ:beta}{b}} \label{equ:betab}\\
&\triangleq\sum\limits_{k\in\mc K}\beta_k \Exp{\left\Vert\mathbf {g}_k^{(n)}-\nabla F_k(\mathbf w^{(n)})\right\Vert^2}, \notag
\end{align}
where \eqref{equ:betaa} can be obtained based on the same reason as obtaining \eqref{equ:alphab}, 
\eqref{equ:betab} is obtained by Lemma 5 in \cite{Zhu_one_bit}, which leads to $\sum\limits_{l=1}^{K}\!\frac 1 {l} C_K^l {q_{s}^l (1-q_{s})^{K-l}}\!\leq\! \frac {2}{Kq_s}$,
parameter $\chi\!=\!1\!-\!(1\!-\!q_{s,min})^K$ and $q_{s,min}$ represents the minimum value of $q_{s,k}(\bs u_k)$.
Therefore, we have $\beta_k=\frac{2}{K^2 \chi q_s}, \forall k\in\mc K$.

\subsection{Expression of $\{\gamma_k\}$}
Due to $q_{s,k}(\bs u_k)=q_s, \forall k\in\mc K$, we can simplify \eqref{equ:M22final} to 
\begin{align}
&\sum\limits_{l=2}^{K}\!\sum\limits_{\mc K_1\cup\mc K_2=\mc K, \atop |\mc K_1|=l, |\mc K_2|=K-l}\!\!\!\! \frac{\prod\limits_{k_2\in\mc K_2}(1\!-\!q_{s,k_2})}{1\!-\!\prod\limits_{k\in\mc K}(1\!-\!q_{s,k})\!-\!\sum\limits_{k\in\mc K}q_{s,k}\prod\limits_{k'\in\mc K, \atop k'\neq k}(1\!-\!q_{s,k'})} 
\cdot \frac {2} {l} \sum\limits_{k_1\in\mc K_1} \left(\left(q_{s,k_1}-\overline q_s\right)^2+\overline q_s^2\right)\Lambda_{k_1}^2 \label{equ:gamma}\\
&=2\sum\limits_{l=2}^{K} C_K^{l} \displaystyle\frac{(1-q_{s})^{K-l}}{1-(1-q_{s})^K-Kq_{s}(1-q_{s})^{K-1}} 
\cdot\frac 1 K\sum\limits_{k\in\mc K} \left(\left(q_{s,k}-\overline q_s\right)^2+\overline q_s^2\right)\Lambda_{k_1}^2 \tag{\ref{equ:gamma}{a}} \label{equ:gammaa}\\
&=2 \cdot\displaystyle\frac{\sum\limits_{l=2}^{K} C_K^{l} (1-q_{s})^{K-l}}{\sum\limits_{l=2}^{K} C_K^{l} (1-q_{s})^{K-l}q_s^l} 
\cdot\frac 1 K\sum\limits_{k\in\mc K} \left(\left(q_{s,k}-\overline q_s\right)^2+\overline q_s^2\right)\Lambda_{k_1}^2 \tag{\ref{equ:gamma}{b}} \label{equ:gammab}\\
&\leq 2 \cdot\displaystyle\frac{\sum\limits_{l=2}^{K} C_K^{l} (1-q_{s})^{K-l}}{\sum\limits_{l=2}^{K} C_K^{l} (1-q_{s})^{K-l}q_s^K} 
\cdot\frac 1 K\sum\limits_{k\in\mc K} \left(\left(q_{s,k}-\overline q_s\right)^2+\overline q_s^2\right)\Lambda_{k_1}^2 \tag{\ref{equ:gamma}{c}} \label{equ:gammac}\\
&=\frac {2} {Kq_s^K}\sum\limits_{k\in\mc K} \left(\left(q_{s,k}-\overline q_s\right)^2+\overline q_s^2\right)\Lambda_{k_1}^2 \notag\\
&\triangleq\sum\limits_{k\in\mc K} \gamma_k\left(\left(q_{s,k}-\overline q_s\right)^2+\overline q_s^2\right)\Lambda_{k}^2, \notag
\end{align}
where \eqref{equ:gammaa} can be obtained based on the same reason as obtaining \eqref{equ:alphab}, 
\eqref{equ:gammab} is due to $\sum\limits_{l=2}^{K} C_K^{l} (1-q_{s})^{K-l}q_s^l=1-(1-q_{s})^K-Kq_{s}(1-q_{s})^{K-1}$,
and \eqref{equ:gammac} is due to $q_s^l>q_s^K, \forall l\geq 2$.
Therefore, we have $\gamma_k=\frac{2}{K q_s^K}, \forall k\in\mc K$.
Finally, by substituting the expression of $\{\alpha_k\}, \{\beta_k\}, \{\gamma_k\}$ into $G$ bound, we obtain Corollary 1.
\begin{equation} \label{equ:proofGsameq}
\begin{array}{rll}
G\leq\displaystyle\frac {\eta} {K}\sum\limits_{k\in\mc K}\left(\displaystyle\frac{\sigma_k^2}{\delta_k}+\Lambda_k^2\right) 
+\frac{2L\eta^2}{K^2 \chi q_s}\sum\limits_{k\in\mc K}\left(\displaystyle\frac{\sigma_k^2}{\delta_k}+2\Lambda_k^2\right)
+\frac{4L\eta^2}{K q_s^{K-2}}\sum\limits_{k\in\mc K}\Lambda_{k}^2.
\end{array}
\end{equation}

\footnotesize
\renewcommand{\refname}{Reference}
\bibliographystyle{IEEEtran}
\bibliography{IEEEabrv,myfile}

\end{document}